\documentclass[12pt,a4paper]{elsarticle}
\usepackage{ulem}
\usepackage{lineno,hyperref}
\modulolinenumbers[5]

\journal{XXX}

\usepackage{amssymb}
\usepackage{amsmath}
\usepackage{amsfonts}
\usepackage{amsthm}
\usepackage{wasysym}
\usepackage{bbold}
\usepackage{bigints}

\usepackage[ansinew]{inputenc}
\usepackage[T1]{fontenc}
\usepackage{lmodern}
\usepackage{textcomp}
\usepackage{graphicx}

\usepackage[]{color}

\usepackage{fancyhdr}
\usepackage{algorithm}

\usepackage{algpseudocode}

\usepackage{tcolorbox}
\tcbuselibrary{breakable}
\usepackage{caption}
\usepackage{subcaption}

\newcommand{\conditional}[2]{\left\vert\vphantom{\frac{1}{1}}\right.}

\newcommand{\ghosteq}{\phantom{=}~~}
\newcommand{\slfv}{$\boldsymbol{\Lambda}$\textbf{V}}
\newcommand{\eq}[1]{Eq.~\ref{#1}}

\newtheorem{lemma}{Lemma}

\theoremstyle{definition}
\newtheorem{definition}{Definition}
\newtheorem{remark}{Remark}

\newtcolorbox{algbox}[3][]{
  width=\textwidth,colframe=black,colback=white,
  sharp corners, breakable,
  before={\captionof{algorithm}{#2}\label{#3}},
  #1
}

\makeatletter
\renewcommand{\ALG@name}{Procedure}
\makeatother

\begin{document}

\begin{frontmatter}

\title{Rate of coalescence of pairs of lineages in the spatial $\lambda$-Fleming-Viot process}

\author{Johannes Wirtz\corref{mycorrespondingauthor}}
\cortext[mycorrespondingauthor]{Corresponding author}
\ead{jwirtz@lirmm.fr}

\author{St\'ephane Guindon\corref{}}
\ead{sguindon@lirmm.fr}
\begin{abstract}
We revisit the spatial $\lambda$-Fleming-Viot process introduced in \cite{barton:slfv}. Particularly, we are interested in the time $T_0$ to the most recent common ancestor for two lineages. We distinguish between the case where the process acts on the entire two-dimensional plane, and on a finite rectangle. Utilizing a differential equation linking $T_0$ with the physical distance between the lineages, we arrive at simple and reasonably accurate approximation schemes for both cases. Furthermore, our analysis enables us to address the question of whether the genealogical process of the model "comes down from infinity", which has been partly answered before in \cite{veber:slfv}.
\end{abstract}


\end{frontmatter}

\section{Introduction}

The spatial $\lambda$-Fleming-Viot process, first described in \cite{barton:slfv}, is a tool in population genetics to model evolution on a spatial continuum. This stochastic process arises as a natural extension of the processes of the Fleming-Viot type that have become part of population genetics since the 1980s (see \cite{fleming:fv,ethier:fv}). A Fleming-Viot process is a measure-valued process $(\mu_t)_{t>0}$ in continuous time, where  $\mu_t$ is a probability measure on a locally compact probability space $E$ (one writes $\mu_t\in\mathcal{M}_1(E)$, with $\mathcal{M}_1(E)$ denoting the set of probability measures on $E$). The interesting feature of this is that the elements of $E$, the allelic "type space" of the model, can be uncountably infinite. For instance, $E$ could be defined as an interval of the real numbers, and the type of an individual would be given by a real number from that interval, which could be taken to represent quantitative traits such as height or weight \cite{fleming:fv,ethier:fv,ohtakimura:mutation}. Classical population models such as the Wright-Fisher model typically incorporate only finitely or at most countably many different allelic types, and so do their diffusion limits; but it can be shown that under suitable conditions, for instance, a sequence of Wright-Fisher models of size $N$ whose initial distibutions converge to some probability measure $\mu_0$ on $E$, has a (weak) limit in the form of a Fleming-Viot process with values in $\mathcal{M}_1(E)$ \cite{ethier:fv}.\\
It has been shown \cite{ethier:fv,donnelly:slfv} that the approach of modeling populations via Fleming-Viot processes was also robust to incorporating not only classical evolutionary mechanisms like selection, recombination and mutation , but also generalisations of the standard reproduction mechanisms that are present in the Wright-Fisher model. In particular, it is possible to incorporate ``extreme'' reproduction events, in which large portions of the population are replaced by the offspring of one single individual. More precisely, ``extreme'' means that the variance of the number of individuals affected by such a reproduction event tends to infinity as the population size increases. Under suitable conditions, there exist Fleming-Viot-type processes representing the diffusion limits of such population models \cite{griffiths:lfv,birkner:lfv,eldon:psi}. These are the so-called $\lambda$-Fleming-Viot processes; their dual processes are represented by the $\lambda$-coalescents \cite{pitman:coalescent,satigov:coalescent,berestycki:coaltheory}, a class of processes generalizing Kingman's coalescent and enabling more than two lineages to collide at the same time (``$k$-mergers''). In this context, $\lambda$ denotes a probability distribution on $[0,1]$, by which the rate of occurrence of $k$-mergers in a sample genealogy, $k\geq 2$ is determined.\\
The spatial $\lambda$-Fleming-Viot process (from here on, \slfv) is a population model in which the type of an individual is determined on an allelic level as well as by its physical location. Suppose the set of allelic types in the population is given by some set $\mathbb{K}$, and the population inhabits some metrizable two-dimensional object $H$ called \textit{habitat}. The most well-studied cases in the existing literature are $H=\mathbb{R}^2$ (e.g. the seminal \citep{barton:slfv}), $H=\mathbb{T}$ (where $\mathbb{T}$ is a torus, \cite{etheridge:slfvtorus}) and $H=\mathcal{A}$, with $\mathcal{A}$ denoting a rectangular object \citep{guindon:slfv}. The state space space is given by $E=H\times \mathbb{K}$ (hence, an individual is characterized by allelic state and physical location; \cite{veber:slfv}). An infinite number of individuals are located at each location $z\in H$ and  $\rho(z,t)\in\mathcal{M}_1(\mathbb{K})$ denotes the distribution of types at location $z$ and time $t$ (This is referred to as the hight-population density limit in \cite{barton:slfv,barton2013}). The proportions of types at each location is modified over time by randomly occurring events facilitating reproduction and death (``REX events'' \cite{guindon:slfv}). These events occur at points $z\in H$ according to a Poisson process of intensity $\lambda>0$. There are several ways for an event to affect the population; the two most common ones are the following:
\begin{definition}\label{def:slfv}
\hfill\\
\begin{enumerate}
\item Let an event occur at time $t$ and position $z$. A number $r$ is drawn from a distribution with variance $\theta^2<\infty$ and a parental location $w$ is sampled uniformly from the "disc" $B_r(z):=\{y:\|y-z\|<r\}$. Then, a type $k\in\mathbb{K}$ is chosen from $\rho(w,t^-)$. At each position $y\in B_r(z)$, the type distribution after the event is given by 
\begin{equation}\rho(y,t^+)=\delta_{k}u_0+\rho(y,t^-)(1-u_0),\end{equation}
where $\delta_k$ is the Dirac measure on the type $k$. This is called the \textit{disc-based} version of the process \cite{barton:slfv}.
\item Let an event occur at time $t$, position $z$, and choose the parental location $w$ according to the gaussian density 
$$\frac{1}{2\pi\theta^2}\exp\left(-\frac{\|z-w\|}{2\theta^2}\right)\mathrm{d}w$$ Choose a type from $\rho(w,t^-)$ as before. At each position $z'\in H$, the composition of the population after the event becomes 
\begin{equation}\label{eq:gaussian}\rho(y,t^+)=\delta_{k}p(z',z)+\rho(z',t^-)(1-p(z',z)),\end{equation}
with 
\begin{equation}\label{eq:gaussianb}p(z',z):=u_0\exp(-\|z'-z\|^2/(2\theta^2)),\end{equation}
This is called the \textit{gaussian} version of the process \cite{barton:slfv}. 
\end{enumerate}

\end{definition}

For the remainder of this work, we will consider the gaussian model. The parameter $\theta^2$, called \textit{dispersal variance} (or \textit{spatial variance}), is obviously a major determinant of the model. The \textit{mortality} $u_0$ controls the impact of a REX event locally. REX events are ``extreme'' in the sense of Fleming-Viot processes, as the amount of the population that is replaced in a single event is of positive mass. Genealogies sampled from a population evolving in this way are not necessarily reproducible by Kingman's coalescent (for instance, due to multiple mergers). On the other hand, considering the limits of the parameters $\lambda$ and $\theta$, it can be shown that the \slfv converges to a $\lambda$-coalescent under suitable conditions, and in some cases even to Kingman's coalescent \cite{barton2013}. It should be noted, though, that $\lambda$ has a different meaning in the context of the \slfv. Also, note that in \cite{limic:slc}, a slightly different model is considered despite the similarity in terminology.\\
The \slfv~possesses several favourable properties to model evolution in space; duality results and backwards-in-time formulations have been described previously \cite{barton:slfv,veber:slfv}, and it is not subject to the "clumping" issues \cite{felsenstein:torus} observed in the classical Wright-Mal\'ecot model that also describes the evolution of organisms spatially distributed along a continuum \cite{wright:ibd,malecot:heredite}.
One application of this model is the inference of the parameters $\lambda,\theta$ and $u_0$ from geo-referenced genetic data, because these parameters allow an assessment of the speed at which genetic variation disseminates across a given habitat, or also how fast newly reached one can be conquered. For example, in \cite{guindon:slfv}, estimation is conducted using a Markov Chain Monte Carlo approach. This approach relies on an extensive parameter augmentation approach in order to calculate the likelihood of spatial coordinates along a genealogy. Although standard Metropolis-Hastings operators apply here, full Bayesian inference is computationally intensive. Alternative inference approaches, based on pairwise coalescence for instance, are therefore required and motivated the work presented here.\\
More specifically, the purpose of our work is to gain insight on the genealogical process within the \slfv, i.e., the genealogical structure of a sample from a population that evolves according to \slfv-mechanisms, in order to improve the efficiency and precision of Bayesian methods like the one mentioned above. The most basic case of a sample of size $n=2$ has been discussed to some extent in \cite{barton:slfv}. The genealogical space for $n=2$ is entirely described by the \textit{time to coalescence} $T_0$ of the two lineages given their initial distance $d_0$ and the location of their most recent common ancestor $X_{T_0}$. Arguably, $T_0$ is the more significant quantity of the two, since it relates to the amount of variation, as well as to the speed of the reproductive mechanism. In \cite{barton:slfv}, a link between $T_0$ and the probability of identity by descent is utilized to obtain a formula for $T_0$; however, its evaluation requires costly numerical integration on $\mathbb{C}$ and the use of nontrivial functions.\\
In this work, we will take a look at $T_0$ from a different angle and describe computationally feasible ways of approximating it. While this certainly will have to be extended to larger sample sizes in the future (with which the presented methodology may be helpful), even analyses based upon pairs of samples have proven to be effective in practice (e.g. the whole literature on Tajima's $D$ \cite{tajima:d}, or, more recently, \cite{barroso:reclandscape}). Indeed, being able to describe the situation for $n=2$ already enables a statistical assessment of geo-referenced genetic data. Additionally, our approach to this problem sheds some light on other features of the process; for instance, we can answer the question whether the \slfv~"comes down from infinity" negatively (which is in line with a similar result obtained in \cite{veber:slfv}.\\
In sections 2  we will review the dynamics of the \slfv. We will consider the case $H=\mathbb{R}^2$ as well as $H=\mathcal{A}$, since the latter relevant from a practical point of view. Afterwards, we will describe the distance process (denoted by $(Z_t)_{t\geq 0}$) between two lineages backward in time under the \slfv~dynamics. Importantly, we will see that $Z_t$ is linked to the distribution of $T_0$ via its moments. Section 5 is devoted to describing numerical approximations.

\section{\textnormal{\slfv}~Dynamics}\label{sec:dynamics}

We first take a look at the model in a finite-habitat setting. We will see that in letting the habitat size tend to infinity, one naturally recovers the original model of \cite{barton:slfv}. 
As per usual with coalescent processes, we consider that time is running backward, i.e. $t>0$ corresponds to a point in time $t$ units of time in the past compared to the origin (where $t=0$). The habitat is defined by a rectangle $\mathcal{A}$ of width $w$ and height $h$. $|\mathcal{A}|=w\cdot h$ denotes its area. Let $\lambda$ denote the intensity of a Poisson process governing the frequency at which REX events take
place. In a time interval of length $h$, there is a probability 

\begin{equation}
\Pr(N_h=k|\alpha)=\frac{\alpha^k}{k!}\exp(-\alpha),
\end{equation}
with $\alpha := \lambda |\mathcal{A}|h$, that the  number $N_h$ of REX events on $\mathcal{A}$ is $k$.\\
The center of a REX event (denoted by the random variable $Z$) is uniformly distributed on
$\mathcal{A}$, i.e., the density of $Z$ is $p_Z(z) = 1/|\mathcal{A}|$. We have
\begin{equation}\label{eq:poisson1}
\lim_{h\rightarrow 0}\Pr(N_h=1)h^{-1}=\lambda|\mathcal{A}|
\end{equation} 
and 
\begin{equation}\label{eq:poisson2}
\lim_{h\rightarrow 0}\Pr(N_h=k)h^{-1}=0
\end{equation}
for $k>1$. The quantity on the right-hand side of \eq{eq:poisson1} is the rate of events of the process, i.e., the waiting times between events are exponentially distributed with parameter $\lambda|\mathcal{A}|$. If a REX event occurs at some position $z\in\mathcal{A}$ and time $t\geq 0$,  the spatial composition of the population is altered in the way described in Eq.~\ref{eq:gaussian}.

Below, we introduce new notations and fundamental quantities that will be used throughout the article.
\begin{definition}\label{def:lineages}

\begin{enumerate}
\item A \textit{lineage} $X=(X_t)_{t\geq 0}$ is the stochastic process of the location $X_t\in\mathbb{R}^2$ of the ancestor of an individual located at $x_0$ in the present, that lived $t$ units of time in the past.
\item Given two lineages $X,Y$ at an initial distance $d_0$, let the random variable $T_0$ denote the time at which $X$ and $Y$ coalesce, i.e. $T_0\in(0,\infty]$. 
\item For two lineages $X,Y$, we let 
$$D_t:=\|X_t-Y_t\|^2$$
denote the random variable describing the squared euclidean distance between the lineages at time $t$. For $T_0\leq t$, we define $D_t=0$.
\item Similarly, we let
$$Z_t:=\|X_t-Y_t\|^2/(4\theta ^ 2)$$
denote the distance between the lineages at time $t$ standardized with respect to the rate of dispersal.
\item Occasionally, we will denote by $\mathcal{Z}_t$ the random variable obtained by "conditioning" $Z_t$ on $T_0>t$. Formally, let $f_{Z_t}(x)$ denote the density of the random variable $Z_t$ evaluated at $x$. Then $\mathcal{Z}_t$ has density
$$
f_{\mathcal{Z}_t}(x)=\begin{cases}
f_{Z_t}(x)/\Pr(T_0>t\conditional~{} d_0) & x>0\\
\delta_0\cdot \Pr(Z_t=0 \cap T_0>t\conditional~{} d_0)/\Pr(T_0>t\conditional~{} d_0) & x=0
\end{cases}
$$
where $\delta_0$ denotes a Dirac measure at $0$.
\end{enumerate}
\end{definition}
Consider a lineage $X$ located at $X_t=x_t\in\mathcal{A}$ at time $t>0$. A ``jump'' of $X$ is its movement to an updated (``older'') ancestral position when it is affected (``hit'') by an event. Assume that an event takes place at time $t$. The probability that $X$ is hit by the event can be obtained as follows:
\begin{equation}\label{eq:rect-hit}
\Pr(X\textnormal{ hit by the event})=\frac{1}{|\mathcal{A}|}\int_{z \in \mathcal{A}}u_0 \exp\big(-\|z-x_t\|^2/2\theta^2\big) \mathrm{d}z ,
\end{equation}
integrating the right-hand side of Eq.~\ref{eq:gaussianb} over all possible locations $z$ for the event center. Since $\mathcal{A}$ is finite, the integral can be calculated, although it involves the error function.\\
Given two lineages $X$ and $Y$ with locations $X_t=x_t$ and $Y_t=y_t$, the probability that they are hit by the same REX event, i.e., the probability that they coalesce, is obtained as follows:
\begin{equation}\label{eq:rect-coal}
\Pr(X,Y\textnormal{ hit by the event})=\frac{1}{|\mathcal{A}|}\int_{z \in \mathcal{A}}u_0^2\exp\left(-\frac{\|x_t-z\|^2+\|y_t-z\|^2}{2\theta^2}\right) \mathrm{d}z
\end{equation}
Making use of \eq{eq:rect-hit}, we can calculate the rate $\rho_{X}$ at which lineage $X$ located at $X_t=x_t$ gets hit by a REX event: 
\begin{align}
\notag\rho_{X}&=\lim_{h\rightarrow 0}\Pr(x_t\textnormal{ hit by any event in an interval of length }h)\cdot h^{-1}\\
\notag&=\lim_{h\rightarrow 0}\Pr(N_h=1)h^{-1}\frac{1}{|\mathcal{A}|}\int_{z \in \mathcal{A}}u_0 \exp\big(-\|z-x_t\|^2/2\theta^2\big) \mathrm{d}z\\
&=\lambda \int_{z \in \mathcal{A}}u_0 \exp\big(-\|z-x_t\|^2/2\theta^2\big) \mathrm{d}z
\end{align}
because of \eq{eq:poisson1} and \eq{eq:poisson2}.
Similarly, the rate of coalescence between two lineages located at $x_t$ and $y_t$ is obtained as follows:
\begin{equation}
\label{eq:coalrectangle}\rho_{X\wedge Y}=\lambda \int_{z \in \mathcal{A}}u_0^2\exp\left(-\frac{\|z-x_t\|^2+\|z-y_t\|^2}{2\theta^2}\right) \mathrm{d}z
\end{equation}

We now consider $\mathbb{H}=\mathbb{R}^2$. Again, REX events are generated by a Poisson point process of intensity $\lambda>0$. This means that on any Borel set $U\subset \mathbb{R}^2$ of finite measure $|U|$, the number of REX events encountered on $U$ in an time interval of length $h$ is Poisson-distributed with parameter $\lambda|U|h$. If at time $t$, a lineage $X$ is located at position $x_t\in\mathbb{R}^2$, the rate at which a REX event appears and affects this lineage can be calculated as the limit of the same rate on a rectangle, letting its size tend to infinity  (we  write $\lim_{|\mathcal{A}|\rightarrow\infty}$ assuming that both $w$ and $h$ become infinite):
\begin{align}
\notag&\rho_{X}^*=\lim_{|\mathcal{A}|\rightarrow\infty}\lambda \int_{z \in \mathcal{A}}u_0 \exp\big(\text{-}\|z-x_t\|^2/2\theta^2\big) \mathrm{d}z\\
\notag&=\lambda\int_{\mathbb{R}^2}u_0 
\exp\big(\text{-}\|z-x_t\|^2/2\theta^2\big) \mathrm{d}z\\
&=2\pi\theta^2 u_0\lambda \\
\label{eq:jumprectangle}&:=\Delta \lambda,
\end{align}
where $z$ denotes the location of the event.\\
When the ancestral lineage located at $x_{t^-}$ is hit by an event taking place at $z$ and time $t$, it changes its position to  $X_{t^+}$. This random variable is approximately distributed as a bivariate normal with mean $z$, and covariance matrix $\theta^2 \mathbf{I}$. As noted in \citep{barton2013}, the normal approximation becomes exact in the limit of high population density, which is the case that we are considering in the present study.\\
The ancestral process of a single lineage can be thought of as that of a particle on the plane changing position according to a Poisson process. If two lineages $X$ and $Y$ are considered, then they both move through the plane with the possibility of a coalescence. The coalescence rate can again be retrieved as the coalescence rate  on the rectangle in the limit of an infinite size:
\begin{align}\label{eq:coalrate}
\rho_{X\wedge Y}^*&=\pi\theta^2 u_0^2\lambda \exp\left(-\frac{\|x_t-y_t\|^2}{4 \theta^2} \right)\\
&=\frac{\Delta\lambda}{2}u_0\exp\left(-\frac{d_t}{4 \theta^2} \right)
\end{align}
where $d_t$ is the value of the squared distance $D_t$ between $X$ and $Y$ at time $t$. Note that we could also replace $\frac{d_t}{4 \theta^2}$ by $z_t$ in this and the following formulae, where $z_t$ is the value taken by the random variable $Z_t$ (see Definition~\ref{def:lineages}). In any case, on $\mathbb{R}^2$, the rate of coalescences between two lineages at some time $t$ depends only on the squared euclidean distance $D_t$ rather than the locations themselves.\\
While each of the two lineages is hit at a rate  $\Delta\lambda $, the total rate of events (hitting either $X$ or $Y$) is
\begin{equation}\label{eq:totalrate}
\rho^*_{X\vee Y} = 2\Delta\lambda\left(1-\frac{u_0}{4}\exp\left(-\frac{d_t}{4 \theta^2}\right)\right)
\end{equation}
The reason for this is that if we were to simply add the rates associated with each lineage, we would put double weight on the events affecting both lineages simultaneously, i.e., the coalescences.\\
One should also note that even though the distribution of $X_{t^+}$ posterior to a REX event of center $z$ is normal with variance $\theta^2$, this is not the case if we consider the distribution of $X_{t^+}$ posterior to a REX event and assuming that the same event did \textit{not} hit $Y$. The probability density of $X_{t^+}$ in such a case can be written down as follows:
\begin{align}\label{eq:nocoal-conditional}
\notag &p_{\{X_{t^+}\conditional~{}!Y\}}(w)\\
=&\bigint_{\mathbb{R}^2}\frac{\exp\left(\text{-}\frac{\|z-x_t\|^2}{2\theta^2}\right)-\exp\left(\text{-}\frac{\|z-x_t\|^2+\|z-y_t\|^2}{2\theta^2}\right)}{\Delta^2\left(1-\frac{1}{2}\exp\left(\text{-}\frac{d_t}{4 \theta^2}\right)\right)}\exp\left(\text{-}\frac{\|z-w\|^2}{2\theta^2}\right)\mathrm{d}z
\end{align}
where, with a slight abuse of notation, we signify by writing $!Y_t$, that the event does not affect the lineage $Y$. Inspection of the right-hand side reveals that the exponential function involving $d_t$ is only of significant magnitude in comparison to the leading term if $x_t$ and $y_t$ are close, or if $\theta^2$ is large. If on the other hand $\theta^2$ is small or the lineages can be assumed to be sufficiently distant from each other, the conditional distribution above remains well-approximated by a normal.\\
The \slfv~on the rectangle and on the plane are very similar in the initial stages. As time progresses, boundary effects come into play on the rectangle, while lineages can expand indefinitely on $\mathbb{R}^2$. This is illustrated in Figure~\ref{fig:trajectories}, where the average squared euclidean distance between two lineages is shown for both cases as time progresses.
Also, on the rectangle, one can see that the average distance approaches some equilibrium value, while in the plane, the distance seems to grow almost linearly. This will be verified in Section~\ref{sec:asymp}.\\
Since $T_0$ on $\mathbb{R}^2$ depends only on $D_t$, it suffices to simulate $D_t$, or alternatively, the "standardized" $Z_t$. Figure~\ref{fig:trajectories} depicts 10000 such trajectories of $Z_t$ on $\mathbb{R}^2$. Individual trajectories resemble Brownian motions with a drift term of strength $2\Delta$.
\begin{figure}
\includegraphics[scale=0.75]{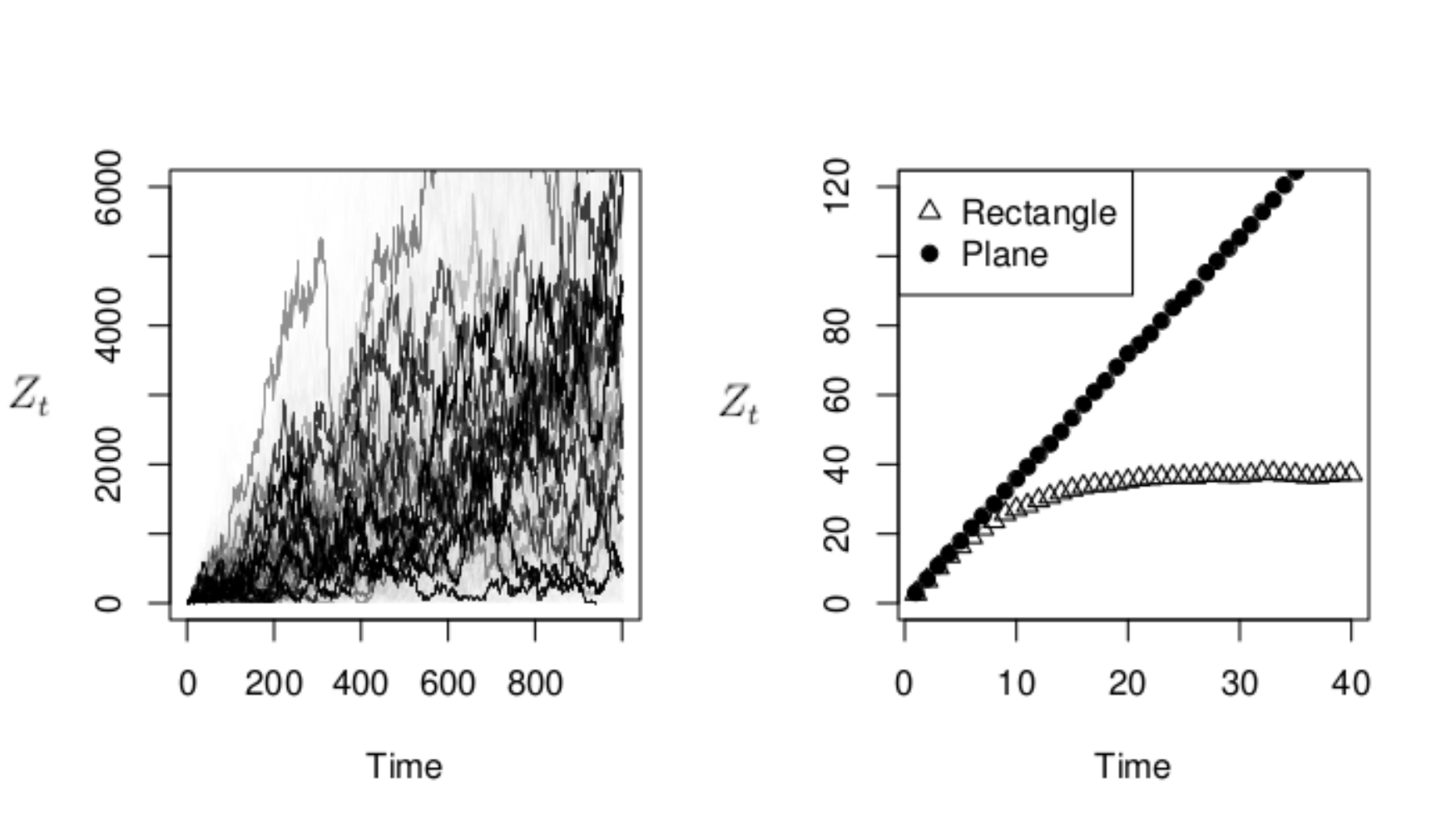}
\caption{Left: Trajectories of $Z_t$ on $\mathbb{R}^2$, using $\lambda=1,u_0=1,\theta^2=1/4,Z_0=0$, over 1000 units of simulated time. Right: the average $Z_t$ over time on $\mathbb{R}^2$ displays a linear trend, on the rectangle it approaches an equilibrium value. In both cases, trajectories undergoing a coalescence were ruled out when generating the plots.}
\label{fig:trajectories}
\end{figure}
We may approximate the cumulative distribution of $T_0$ from such simulations of $D_t$ (see Figure~\ref{fig:coalprob}). Increasing the dispersal variance $\theta^2$ seems to accelerate the process in the long term (however, this effect is different from increasing $\lambda$, which acts as a scaling parameter). Changing the initial distance $d_0$ appears to affect the limiting probability of coalescence $1-p^*$. For large $t$, the distributions seem to run in parallel. Sections \ref{sec:plane} and \ref{sec:approx} will be devoted to finding numerical approximations of these curves.

\begin{figure}
\includegraphics[scale=.75]{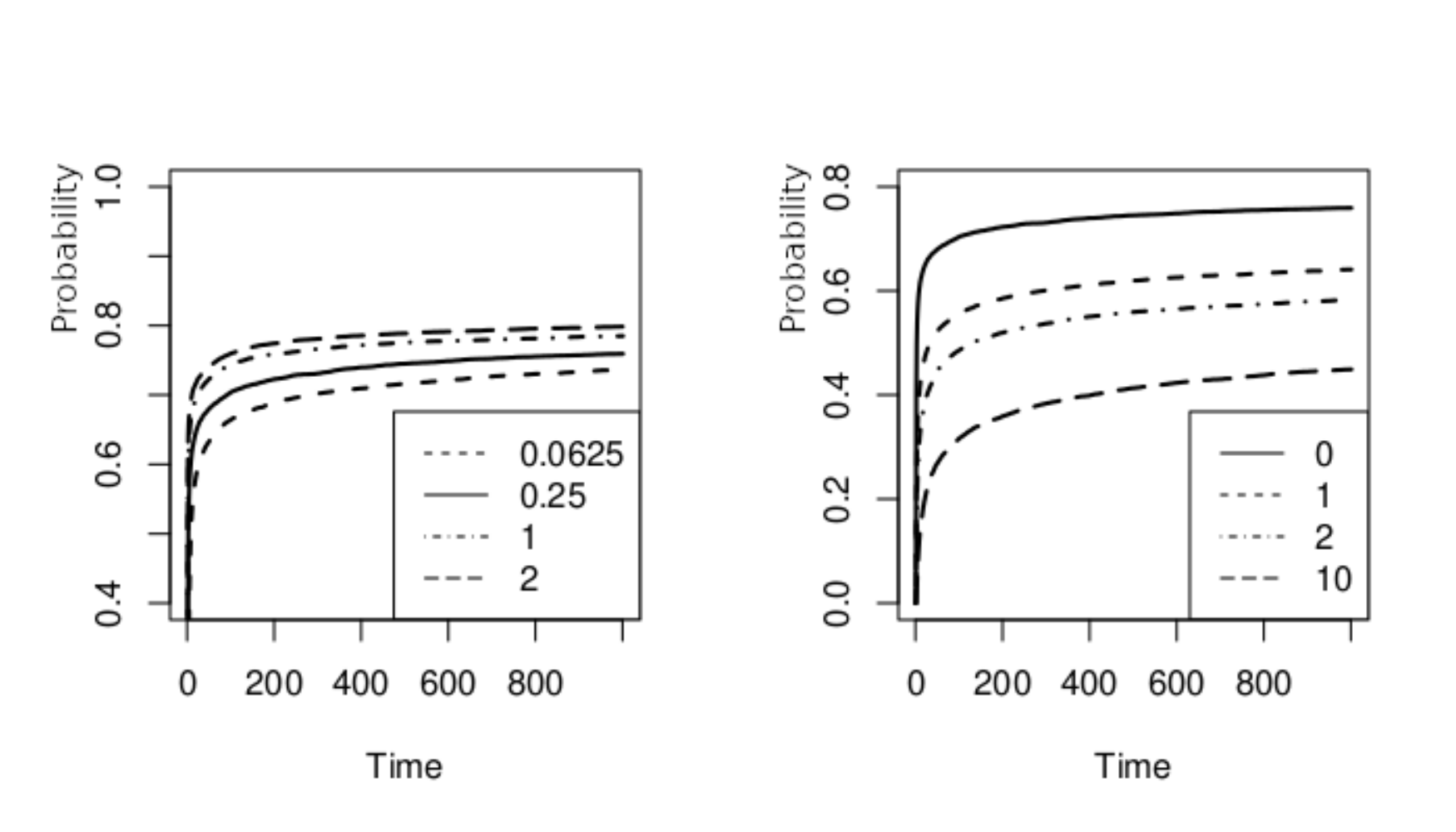}
\caption{\textit{Left:} Coalescence probability for varying values of $\theta^2$. Increasing $\theta^2$ has a similar effect as speeding up time. According to these simulations, $p^*\approx 0.2$ for $d_0=0$. \textit{Right:} Coalescent probability for $\theta^2=1/4$ and $d_0=0,1,2,10$. Again, $\lambda=1$ and $u_0=1$ were used.}
\label{fig:coalprob}
\end{figure}

\section{Derivation of the coalescence probability}\label{sec:coaltime}
In the following, we describe two ways of approaching the distribution of $T_0$ by symbolical methods. We first assume that the habitat is $\mathbb{R}^2$.
\subsection{The distribution as a solution to an \textit{ODE}} 

We consider the probability that coalescence between the two ancestral lineages that started at distance $d_0$ takes place before or at time $t+h$. We may write

\begin{eqnarray*}
&~&\Pr(T_0\leq t+h|d_0)\\
&=&\Pr(T_0\leq t|d_0)+\Pr(t< T_0\leq t+h|d_0)\\
&=&\Pr(T_0\leq t|d_0)+\Pr(T_0\leq t+h|d_0,T_0>t) (1-\Pr(T_0\leq t|d_0))
\end{eqnarray*}

which, by rearranging and considering the limit of $h\rightarrow 0$, can be transformed into

\begin{equation}\label{eq:coalode}\frac{\frac{\partial}{\partial t}\Pr(T_0\leq t|d_0)}{1-\Pr(t\leq T_0|d_0)}=\lim_{h\rightarrow 0}\frac{1}{h}\Pr( T_0\leq t+h|d_0,T_0> t)\end{equation}

Integrating and taking exponentials, we arrive at
\begin{equation}\Pr(T_0\leq t|d_0)=1-\exp\left(-\lim_{h\rightarrow 0}\frac{1}{h}\int_0^{t}\Pr(T_0\leq u+h|d_0,T_0>u))\mathrm{d}u\right)\label{eq:1strep}\end{equation}
Let $f_u(\cdot|T_0>u,d_0)$  denote the density of $D_u$, conditioned on $T_0>u$ and the distance at $t=0$ being equal to $d_0$. Then we have:\\
\begin{equation}\label{eq:pre-ode}
\Pr(u< T_0\leq u+h|d_0,T_0>u)=\int_0^{\infty}\Pr(T_u\leq u+h|d_u=x)f_u(x|T_0>u,d_0)\mathrm{d}x
\end{equation} 
The limit in \eq{eq:1strep} may be moved inside the integral, and $\lim_{h\rightarrow 0}\frac{1}{h}\Pr(T_0\leq h|d_0=x)$ is given by the right-hand side of equation~\ref{eq:coalrate}. Therefore, we have

\begin{eqnarray}\label{eq:coalrate-ode}
\Pr\left(T_0\leq t|d_0\right)=1-\exp\left(-\frac{1}{2}\Delta \lambda u_0\int_0^{t} \mathbb{E}\left(\exp\left(-Z_u\right)\conditional~{} T_0>u,d_0\right)\mathrm{d}u\right)
\end{eqnarray}
so that
\begin{eqnarray}
\log\Pr(T_0 > t|d_0)=-\frac{1}{2}\Delta \lambda u_0 \left(\int_0^{t} \mathbb{E}\left(\exp\left(-Z_u\right)\conditional~{} T_0>u,d_0\right)\mathrm{d}u\right)
\end{eqnarray}
It is worth pointing out that one may exchange the integrals on the right-hand side, allowing us to take the integral with respect to $t$ over the density of $Z_t$ alone:

\begin{eqnarray*}
\int_0^{t} \mathbb{E}\left(\exp\left(-Z_u\right)\conditional~{} T_0>u,d_0\right)\mathrm{d}u &=& 
\int_0^{t} \int_{0}^{\infty} \exp(-x) p_{Z_u}(x|T_0>u,d_0) \mathrm{d}x \mathrm{d}u \\
&=&\int_{0}^{\infty} \int_0^{t}  \exp(-x) p_{Z_u}(x|T_0>u,d_0) \mathrm{d}u \mathrm{d}x \\
&=&\int_{0}^{\infty} \exp(-x) \int_0^{t}  p_{Z_u}(x|T_0>u,d_0) \mathrm{d}u \mathrm{d}x
\end{eqnarray*}
Hence, the probability that coalescence takes place before a given point in time given the initial distance between the two lineages can be understood as the integral, taken over the time period considered, of the moment-generating function of the standardized distance (term to the left of the equality sign in the equation above). It can also be understood as the moment-generating function of the random variable $\mathcal{Z}_t$ (see Definition~\ref{def:lineages}).

\subsection{$T_0$ as a Cox Process}
In this section, we provide a different take on the same problem by decomposing the \slfv~into two stages: 
\begin{enumerate}
\item In the first stage, $X$ and $Y$ (with fixed $x_0, y_0$) move across $\mathbb{R}^2$, with the rate of events given by $\rho_{X\vee Y}^*-\rho_{X\wedge Y}^*$; every event affects either $X$ or $Y$ and \eq{eq:nocoal-conditional} is used to update the positions of lineages (rather than a normal density). This induces a random path of the squared euclidean distance between them in $\mathbb{R}_0^+$, i.e. a random piecewise-constant function $D'_t:=\|X_t-Y_t\|^2, t\in\mathbb{R}_0^+$ 

\item Along $\mathbb{R}_0^+$, "potential coalescent events" are distributed according to a non-homogenous Poisson process (i.e., a Cox Process \cite{cox:process}) with rate function 
\begin{equation}\rho(t):=\frac{1}{2}\Delta \lambda u_0 \exp\left(-\frac{D'_t)}{4\theta^2}\right)\end{equation}
The first potential coalescent event encountered along $\mathbb{R}_0^+$ finally represents the actual coalescent event of the lineages.
\end{enumerate}
The rates at which lineages change locations or coalesce are equal to the rates under the \slfv; therefore the above is an equivalent description of the process for two lineages. Obviously, we may also simulate the process in this way.\\
Given a trajectory $\delta(t)$ of distances between lineages over time, the probability distribution of the time $T_0$ until we encounter a coalescence is given by
\begin{equation}\label{eq:coaldist-pathwise}\Pr(T_0\leq t|\{\delta(u),0\leq u \leq t\})=1-\exp\left(-m(t)\right),\end{equation}
where $m(t) := \int_0^t \rho(u)\mathrm{d}u$.
The probability distribution of $T_0$ under the \slfv~thus equals equation~\ref{eq:coaldist-pathwise} averaged over all possible paths $\delta|_{[0,t]}$ between $0$ and $t$, i.e.
\begin{equation}\label{eq:coaldist-pathwise-avg}\Pr(T_0\leq t|d_0)=1-\mathbb{E}\left(\exp\left(-m(t)\right)\right)\end{equation}
Note that the random variable of this expression is $-\int_0^t\frac{1}{2}\Delta \lambda u_0 \exp\left(-\frac{\delta(u)}{4\theta^2}\right)\mathrm{d}u$, so the expectation here is taken over time as well as over space.

\begin{remark}
The similarity between equations~\ref{eq:coalrate-ode} and~\ref{eq:coaldist-pathwise-avg} suggests that one may interchange expectation and exponential. Note however that the expectation in~\eq{eq:coaldist-pathwise-avg} is taken over all paths generated in the first stage of the \slfv, where coalescence events are not taken into account, whereas in~\eq{eq:coalrate-ode} the expectation is conditioned on coalescence events {\it not} taking place up to $u$ in the original process.
\end{remark}

\subsection{The coalescence process on a rectangle}
Now, we consider again the case where the habitat is given by a rectangle $\mathcal{A}$. We can, in fact, derive a slightly modified version of \eq{eq:coalrate-ode}. Here, the coalescence probability depends on the lineage position relative to the border of $\mathcal{A}$, which is why one needs to condition on $X_t$ and $Y_t$ (rather than $D_t$ or $Z_t$). More precisely, \eq{eq:pre-ode} becomes
\begin{align}
&\ghosteq\Pr(u< T_0\leq u+h|x_0,y_0,T_0>u)\\
\notag&=\int_0^{\infty}\Pr(T_u\leq u+h|X_u=x,Y_u=y,T_0>u,x_0,y_0)\\
\notag&\ghosteq \cdot f_u(x,y|T_0>u,x_0,y_0)\mathrm{d}(x,y)
\end{align}
where $f_u(X_u,Y_u|T_0>u,x_0,y_0)$  denotes the joint density of $X_u,Y_u$ conditioned on $T_0>u$ and $x_0,y_0$, and we can again move the limit into the integral and evaluate it, with \eq{eq:coalrectangle} substituted for \eq{eq:coalrate}. The result is
\begin{align}
\notag&\ghosteq\log\left(\Pr(T_0> t|x_0,y_0)\right)\\
\label{eq:coalrate-ode-rect}&=-\int_0^{t}\mathbb{E}\left(\lambda u_0^2 \exp\left(-\frac{\|X_u-Z\|^2+\|Y_u-Z\|^2}{2\theta^2}\right)\conditional~{} T_0>u,x_0,y_0 \right)\mathrm{d}u
\end{align}
where $Z$ is an event location uniformly distributed on $\mathcal{A}$, and the expectation is taken over $X_u$, $Y_u$ and $Z$.\\
Since the conditional expectation $\mathbb{E}\left(\exp\left(-\frac{\|X_t-Z\|^2+\|Y_t-Z\|^2}{2\theta^2}\right)\conditional~{} T_0>t,x_0,y_0 \right)$ is bounded, it has to approach a limiting value $c$, with $1\geq c>0$ as $t \to \infty$. If $t$ is large,  we have the following approximation for the density of coalescence times:
\begin{equation}\label{eq:coalrate-rectangle}
\frac{\partial}{\partial t}\Pr(T_0\leq t|x_0,y_0)\propto\exp\left(-\lambda u_0^2ct\right)
\end{equation}
For large $t$, the density is thus proportional to that of an exponentially distributed random variable with parameter $\lambda u_0^2c$. 
More generally, the joint distribution of $X_t$ and $Y_t$ conditioned on $T_0>t$ approaches a quasi-stationary distribution \cite{bartlett:popmodels}. Simulations suggest that it resembles the uniform distribution on $\left(\mathcal{A}\times \mathcal{A}\right)$.\\
While one can evaluate the value of $c$ numerically, we point out that it is also approximated by the equivalent term of \eq{eq:coalrate-ode}, i.e.
$$c\approx\frac{\Delta}{2}\mathbb{E}\left(\exp\left(-Z_t\right)\conditional~{} T_0>t,x_0,y_0 \right)$$
which relates $c$ it to the adjusted distance $Z_t$. Assuming $X_t,Y_t$ are independent and uniformly distributed on $\mathcal{A}$, and making use of a result presented in \cite{philip2007}, the distribution of $Z_t$ is given by
\begin{eqnarray*}
p_{Z_{t}}^*(x|T_0>t,d_0) = 
4\theta^2 \times 
\left\{
\begin{array}{ll}
-2 \frac{\sqrt{d}}{w^2h} - 2 \frac{\sqrt{d}}{wh^2} + \frac{\pi}{wh} + \frac{d}{w^2h^2},& \\ 
\hspace{0.35\textwidth} \text{if} \quad 0 < d \leq w^2 & \\
&\\
-2 \frac{\sqrt{d}}{w^2h} & \\ 
- \frac{1}{h^2} + \frac{2}{wh}\arcsin{\frac{w}{\sqrt{d}}}+\frac{2}{w^2h} \sqrt{d-w^2}, & \\
\hspace{0.35\textwidth} \text{if}\quad w^2 < d \leq h^2 &\\
& \\
- \frac{1}{h^2} + \frac{2}{wh}\arcsin{\frac{w}{\sqrt{d}}}+\frac{2}{w^2h} \sqrt{d-w^2} & \\
- \frac{1}{w^2} + \frac{2}{wh}\arcsin{\frac{h}{\sqrt{d}}}+\frac{2}{wh^2} \sqrt{d-h^2} & \\
-\frac{\pi}{wh} - \frac{d}{w^2h^2},& \\
\hspace{0.35\textwidth} \text{if}\quad h^2 < d \leq w^2+h^2 & 
\end{array}\right.
\end{eqnarray*}
with $d := 4\theta^2x$ and without restriction $w\leq h$. $c$ is then approximated by
$$\mathbb{E}\left(\exp\left(-Z_t\right)\conditional~{} T_0>t,x_0,y_0 \right)=\int_{-\infty}^{0} \exp(x) \int_0^{t}  p_{Z_t}(x|T_0>t,d_0) \mathrm{d}u \mathrm{d}x$$
and density of $Z_t$ is approximated by
\begin{equation}
p_{Z_t}(x|T_0>t,d_0) \approx \delta(z_0-x) \exp(\alpha t) + p_{Z_{t}}^*(x|T_0>t,d_0) \left(1-\exp(\alpha t)\right)
\end{equation}
where $\alpha$ is the probability that an event with uniformly chosen location on $\mathcal{A}$ neither affects lineage $X$ located at $x_0$ nor $Y$ at $y_0$. In other words, $\exp(-\alpha t)$ is the probability that at time $t$ the lineages $X,Y$ are both located at their initial positions $x_0,y_0$. It can be calculated using the formulae in Section \ref{sec:dynamics}.
If $X_0$ and $Y_0$ are themselves uniformly sampled from $\mathcal{A}$, the natural approximation for $\Pr(T_0\leq t)$ is $1-\exp(-\lambda u_0^2 c t)$. If $x_0$ and $y_0$ are provided, we propose
\begin{equation}
\Pr(T_0\leq t|x_0,y_0) \approx 1-\exp\left(-\lambda u_0^2t\left(\exp(-\alpha t)\frac{\Delta z_0}{2}+\left(1 - \exp(-\alpha t)\right)c\right)\right)
\end{equation}
with $\alpha$ defined as above.\\
Generally, the results we obtain suggest that the approximation proposed here works best if the rectangle is not too large in relation to $\theta^2$. Otherwise, the fact that the quasi-stationary distribution is not exactly uniform seems to negatively affect the accuracy of determining $c$. The quasi-stationary distribution can also be found as the solution to a functional equation, but seems difficult to approach numerically.
\begin{figure}
\includegraphics[scale=0.75]{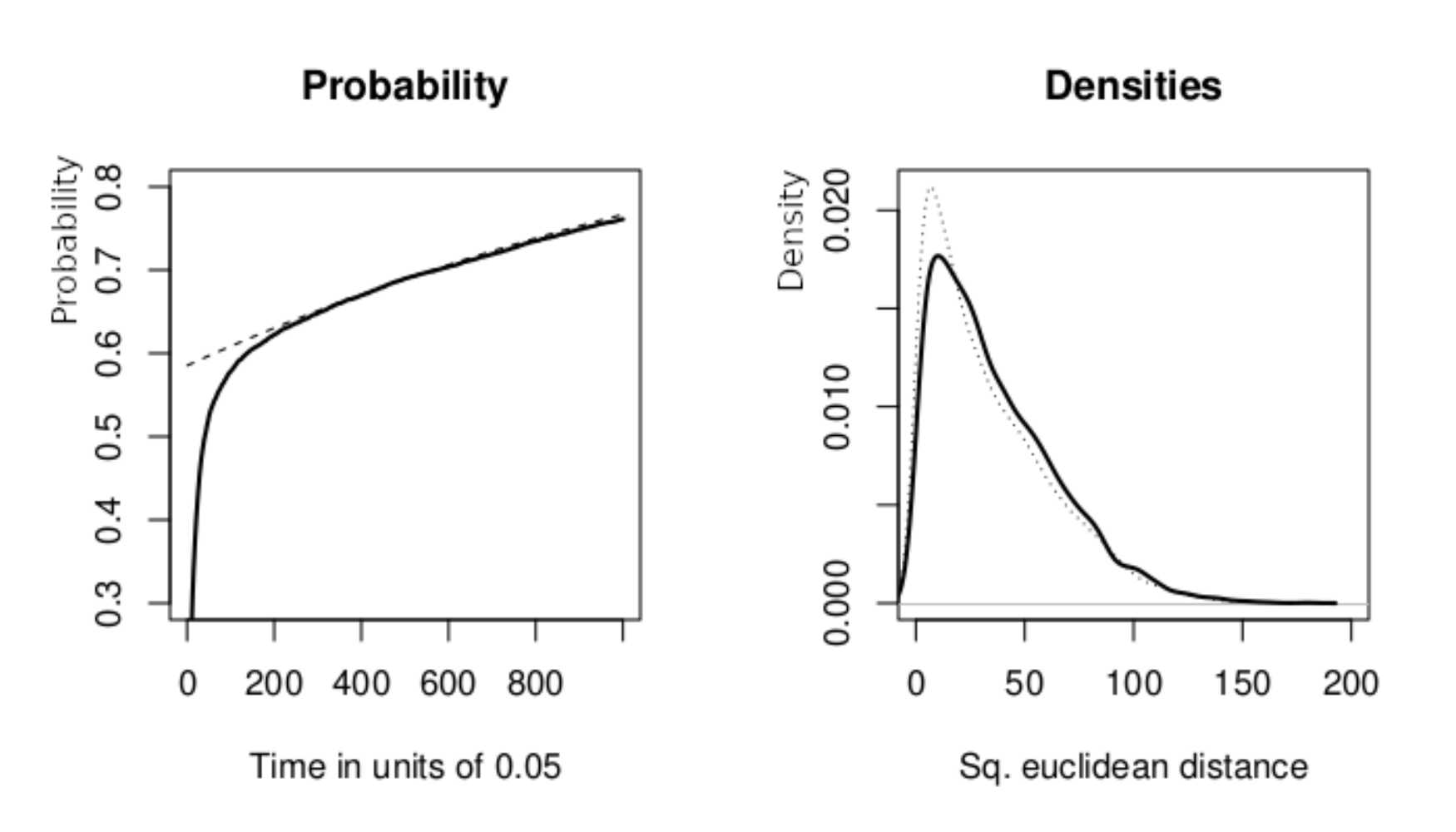}
\caption{\textit{Left}: Probability of coalescence times on a $10\times 10$ rectangle, $\theta^2=1/4$, $\lambda=1$, for fixed initial distances $0$. The value of $c$ is around $0.012$. Shown in dots is the function $1-\exp(-c(t+t_0))$ (an exponential distribution with parameter $c$, shifted to the left by $t_0\approx 0.8$). The bigger $t$ gets, the more similar the two curves become. \textit{Right}: The quasi-stationary distribution of $Z_t$ on the rectangle (black line) resembles the distribution of $Z_t$ when both positions are uniform (dots), with a slight preference for higher values.}
\label{fig:densitiesrectangle}
\end{figure}

\section{The dynamics of $Z_t$}\label{sec:plane}
The distribution of coalescence times on a rectangle could be approximated by relatively simple symbolical methods; to achieve the same for $\mathbb{R}^2$, analyzing the process of $Z_t$, given the initial distance $d_0$, turns out to be instructive. We assume $u_0=1$ from here on; the results may be reproduced in the same way for other values of this parameter. Consider the following function:
\begin{equation}\label{eq:mdef}
\mathcal{M}_{Z_t,d_0}(s):=\mathbb{E}\left(\exp\left(-sZ_t\right)\mathbb{1}_{T_0>t}\conditional~{} d_0\right)
\end{equation}
Usually, we will denote this function by $\mathcal{M}_{Z_t}(s)$ for short, unless we want to consider multiple values of $d_0$. One can think of this function as the moment-generating function of the random value $Z_t$, multiplied by an additional indicator function that returns $0$ if the pair of lineages has coalesced at time $t$ and $1$ if it hasn't. This function conveys a lot of information on the \slfv~in general and on coalescence times in particular. For example, we have
\begin{equation}\label{eq:fundamental}
\Pr\left(T_0\leq t\conditional~{} d_0\right)=1-\mathcal{M}_{Z_t}(0)
\end{equation}
Furthermore, it holds that
\begin{equation}\label{eq:moments1}
\left(\frac{\partial}{\partial s}\right)^k\mathcal{M}_{Z_t}(s)=\mathbb{E}\left(\left(-Z_t\right)^k\cdot\exp\left(-sZ_t\right)\mathbb{1}_{T_0>t}\conditional~{} d_0\right)
\end{equation}
for all $k\in\mathbb{N}$, and therefore
\begin{equation}\label{eq:moments2}
\left(\frac{\partial}{\partial s}\right)^k\mathcal{M}_{Z_t}(0)=\mathbb{E}\left(\left(-Z_t\right)^k\mathbb{1}_{T_0>t}\conditional~{} d_0\right)
\end{equation}
Using \eq{eq:fundamental}, we may even write
\begin{equation}
\left(\left(\frac{\partial}{\partial s}\right)^k\mathcal{M}_{Z_t}(0)\right)\cdot\left(\mathcal{M}_{Z_t}(0)\right)^{-1}=\mathbb{E}\left(\left(-Z_t\right)^k\conditional~{} T_0>t,d_0\right)
\end{equation}
In particular, we can compute the expectation appearing in Equation~\ref{eq:coalrate-ode}: 
\begin{align}\label{eq:mgfZ}
&\ghosteq\mathbb{E}\left(\frac{1}{2}\Delta \lambda u_0 \exp(-Z_t)\conditional~{} T_0>t,d_0 \right)\\
\notag&=\frac{\mathbb{E}\left(\frac{1}{2}\Delta \lambda u_0 \exp(-Z_t)\mathbb{1}_{T_0>t}\conditional~{} d_0\right)}{\Pr\left(T_0>t\conditional~{} d_0\right)}\\
\notag&=-\frac{1}{2}\Delta \lambda u_0\frac{\mathcal{M}_{Z_t}(1)}{\mathcal{M}_{Z_t}(0)}
\end{align}
It is possible to obtain a closed form of the derivative of $\mathcal{M}_{Z_t}(s)$ with respect to $t$ for any $s$. After close inspection of $\frac{\partial}{\partial t}\mathcal{M}_{D_t}(s)$ and a number of algebraic manipulations (see Appendix), we arrive at
\begin{align}
\label{eq:coalincorporated}
&\ghosteq\frac{\partial}{\partial t}\mathcal{M}_{Z_t}(s)\\
\notag&=2\Delta\lambda\left(\frac{1}{s+1}\mathcal{M}_{Z_t}\left(\frac{s}{1+s}\right)-\mathcal{M}_{Z_t}(s)-\frac{2}{3s+4}\mathcal{M}_{Z_t}\left(1+\frac{s}{3s+4}\right)\right)\\
\notag&~~+\frac{\Delta\lambda}{2}\mathcal{M}_{Z_t}(1+s)\\
\end{align}

The above is a "partial differential-functional equation" and to our knowledge not analytically solvable. For certain values of $s$, we still gain some insight on the process; for instance, plugging in $0$ for $s$ yields
\begin{equation}\label{eq:delp}
\frac{\partial}{\partial t}\Pr(T_0> t| d_0)=\frac{\partial}{\partial t}\mathcal{M}_{Z_t}(0)=-\frac{1}{2}\Delta \lambda u_0\mathcal{M}_{Z_t}(1)
\end{equation}
Two more crucial features of $\mathcal{M}_{Z_t}(s)$ are:
\begin{lemma}\label{lemma:convergence}
\begin{itemize}
\item[	a)]$\lim_{t\rightarrow \infty}\left(\frac{\partial}{\partial s}\right)^k\mathcal{M}_{Z_t}(s)=0$ for all $s>0,k\geq 0$.
\item[b)]$\lim_{t\rightarrow \infty}\mathcal{M}_{Z_t}(0)=p^*>0$ 
\end{itemize}
\end{lemma}
We defer the proof of this lemma to the appendix.\\
Lemma~\ref{lemma:convergence}\textit{b)} states that lineage pairs are not required to coalesce on $\mathbb{R}^2$. In some cases $Z_t$ never ceases to grow, and consequently the intensity of the coalescent process tends to $0$ so quickly that coalescence never occurs. This result also suggests that the \slfv~does not "come down from infinity", which is to say that if a sample of infinite size is taken from the population, then, looking at the genealogical process of this sample, one will always encounter an infinite amount of lineages ("dust") that have not coalesced up to any time $t$ back in the past. This problem is also treated in \cite{veber:slfv}, where it is proven explicitly that the \slfv~under its disc-based definition (see Definition~\ref{def:slfv}) does not come down from infinity. Despite their differences, it seems only natural to expect a similar statement to hold for the gaussian version.

\section{Approximation of the probability distribution}\label{sec:approx}
In order to approximate the numerical values of $\mathcal{M}_{Z_t}(0)=\Pr(T_0>t)$, we propose an approach that relies on the Taylor expansion of the function of interest, and one exploiting its representation as an ODE. Combining the two, one obtains a good approximation of $\mathcal{M}_{Z_t}(0)$. Further approaches (such as numerically solving \eq{eq:coalincorporated} by a Runge-Kutta scheme) can be envisioned, but either seem less accurate or computationally unfeasible.

\subsection{Calculation of the Taylor expansion}
We consider the Taylor series expansion of $\mathcal{M}_{Z_t}(0)$ at $t_0=0$:
\begin{equation}
\mathcal{M}_{Z_t}(0)=\sum_{j\in\mathbb{N}}\frac{t^j}{j!} \left(\frac{\partial}{\partial t}\right)^j\mathcal{M}_{Z_t}(0)\conditional~{}_{t=0}=\sum_{j\in\mathbb{N}}g_{j}t^j
\end{equation}
with $g_{j}:=\frac{1}{j!}\left(\left(\frac{\partial}{\partial t}\right)^j\mathcal{M}_{Z_t}(0)\right)\conditional~{}_{t=0}$.
We obtain these coefficients by considering the differential equation \eq{eq:coalincorporated} that is solved by $\mathcal{M}_{Z_t}(s)$. Multiple derivation with respect to $t$ yields
\begin{align}\label{eq:mgfcoalderived}
&\ghosteq\left(\frac{\partial}{\partial t}\right)^j\mathcal{M}_{Z_t}(s)\\
\notag&=2\Delta\lambda\left(\frac{1}{s+1}\left(\frac{\partial}{\partial t}\right)^{j-1}\mathcal{M}_{Z_t}\left(\frac{s}{1+s}\right)-\left(\frac{\partial}{\partial t}\right)^{j-1}\mathcal{M}_{Z_t}(s)\right)\\
\notag&~~+2\Delta\lambda\left(-\frac{2}{3s+4}\left(\frac{\partial}{\partial t}\right)^{j-1}\mathcal{M}_{Z_t}\left(1+\frac{s}{3s+4}\right)+\frac{1}{4}\left(\frac{\partial}{\partial t}\right)^{j-1}\mathcal{M}_{Z_t}(1+s)\right)
\end{align}
Setting $s=0$ in \eq{eq:mgfcoalderived}, we obtain, $\forall j > 0$
\begin{align}\label{eq:coeffrecursion}
&\ghosteq g_{j}\cdot j! = \left(\frac{\partial}{\partial t}\right)^j\mathcal{M}_{Z_t}(0)\conditional~{}_{t=0}\\
\notag&=2\Delta\lambda\left(\left(\frac{\partial}{\partial t}\right)^{j-1}\mathcal{M}_{Z_t}\left(0\right)\conditional~{}_{t=0}-\left(\frac{\partial}{\partial t}\right)^{j-1}\mathcal{M}_{Z_t}(0)\conditional~{}_{t=0}\right)\\
\notag&+\Delta\lambda\left(-\left(\frac{\partial}{\partial t}\right)^{j-1}\mathcal{M}_{Z_t}\left(1\right)\conditional~{}_{t=0}+\frac{1}{2}\left(\frac{\partial}{\partial t}\right)^{j-1}\mathcal{M}_{Z_t}(1)\conditional~{}_{t=0}\right)\\
=\notag&-\frac{\Delta\lambda}{2}\left(\left(\frac{\partial}{\partial t}\right)^{j-1}\mathcal{M}_{Z_t}\left(1\right)\right)\conditional~{}_{t=0}
\end{align}
By virtue of \eq{eq:coeffrecursion}, any term of the form $\left(\left(\frac{\partial}{\partial t}\right)^{k}\mathcal{M}_{Z_t}(\sigma)\right)\conditional~{}_{t=0}$, $\sigma>0$, may be expressed by terms of the form $\left(\left(\frac{\partial}{\partial t}\right)^{k-1}\mathcal{M}_{Z_t}(\tau)\right)\conditional~{}_{t=0}$, $\tau>0$. The repeated application of this equation results in an expression of the form
\begin{equation}
g_{j}\cdot j! = \sum_{i=1}^{4^j}\beta_i \left(\left(\frac{\partial}{\partial t}\right)^{0}\mathcal{M}_{Z_t}(\sigma_i)\right)\conditional~{}_{t=0}
\end{equation}
where $\left(\left(\frac{\partial}{\partial t}\right)^{0}\mathcal{M}_{Z_t}(s)\right)\conditional~{}_{t=0}=\mathcal{M}_{Z_0}(s)=\exp\left(-\frac{d_0s}{4\theta^2}\right)$ and $\beta_i\in\mathbb{R},\sigma_i>0$; so ultimately,
\begin{equation}
g^*_{0k}\cdot k! = \sum_{i=1}^{4^k}\beta_i \exp\left(-\frac{d_0\sigma_i}{4\theta^2}\right)
\end{equation}
Let $\gamma^{(J)}(t):=\sum_{j=0}^{J}g_{j}t^j$ denote the $J$-th order Taylor polynomial of $\mathcal{M}_{Z_t}(s)$. These polynomials approximate $\mathcal{M}_{Z_t}(s)$ very well for small values of $t$ (Figure~\ref{fig:taylor}). However, calculating successively higher orders quickly becomes computationally intense. We will therefore consider another strategy of approximating $\mathcal{M}_{Z_t}(s)$, sacrificing accuracy initially in exchange for being able to correctly display the long-term behavior of $\mathcal{M}_{Z_t}(s)$. An additional ingredient we will need is the asymptotic behavior of the first two moments of $Z_t$.

\subsection{Asymptotic of first and second moment}\label{sec:asymp}
Recalling \eq{eq:moments1} and \eq{eq:moments2}, we have
\begin{align}
\frac{\partial}{\partial s}\mathcal{M}_{Z_t}(s)&=\mathbb{E}\left(-Z_t\cdot\exp\left(-sZ_t\right)\mathbb{1}_{T_0>t}\conditional~{} d_0\right)\\
\left(\frac{\partial}{\partial s}\right)^{2}\mathcal{M}_{Z_t}(s)&=\mathbb{E}\left(Z_t^2\cdot\exp\left(-sZ_t\right)\mathbb{1}_{T_0>t}\conditional~{} d_0\right)
\end{align}
and setting $s=0$,
\begin{align}
\frac{\partial}{\partial s}\mathcal{M}_{Z_t}(s)\conditional~{}_{s=0}&=\mathbb{E}\left(-Z_t\mathbb{1}_{T_0>t}\conditional~{} d_0\right)\\
\left(\frac{\partial}{\partial s}\right)^{2}\mathcal{M}_{Z_t}(s)\conditional~{}_{s=0}&=\mathbb{E}\left(Z_t^2\mathbb{1}_{T_0>t}\conditional~{} d_0\right)
\end{align}
Performing the same derivations on the right-hand side of  \eq{eq:coalincorporated} and setting $s=0$ leads to differential equations for the first and second moments of  $Z_t$:
\begin{align}
\label{eq:dtgdl}\frac{\partial}{\partial t}\mathbb{E}\left(Z_t\mathbb{1}_{T_0>t}\conditional~{} d_0\right)&=2\Delta\lambda\left(\mathcal{M}_{Z_t}(0)-\frac{1}{8}\left(3\mathcal{M}_{Z_t}(1)+\frac{\partial}{\partial s}\mathcal{M}_{Z_t}(s)\conditional~{}_{s=1}\right)\right)\\
\notag \frac{\partial}{\partial t}\mathbb{E}\left(Z_t^2\mathbb{1}_{T_0>t}\conditional~{} d_0\right)&=4\Delta\lambda\left(\mathcal{M}_{Z_t}(0)+2\mathbb{E}\left(Z_t\mathbb{1}_{T_0>t}\conditional~{} d_0\right)\right)-\frac{9}{8}\Delta\lambda\mathcal{M}_{Z_t}(1)\\
&+\frac{1}{8}\Delta\lambda\left(6\frac{\partial}{\partial s}\mathcal{M}_{Z_t}(s)\conditional~{}_{s=1}+\frac{7}{2}\left(\frac{\partial}{\partial s}\right)^{2}\mathcal{M}_{Z_t}(s)\conditional~{}_{s=1}\right)
\end{align}
The right-hand side of \eq{eq:dtgdl} is nonnegative, because $\mathcal{M}_{Z_t}(1)\leq\mathcal{M}_{Z_t}(0)$ as well as $\frac{\partial}{\partial s}\mathcal{M}_{Z_t}(s)\conditional~{}_{s=1}\leq\mathcal{M}_{Z_t}(0)$, so unsurprisingly, $\mathbb{E}\left(Z_t\mathbb{1}_{T_0>t}\conditional~{} d_0\right)$ is monotonously increasing. Also, because $\mathcal{M}_{Z_t}$ and all its derivatives with respect to $s$ evaluated at $s>0$ vanish as $t\rightarrow \infty$ (Lemma~\ref{lemma:convergence}), we have the following approximation for large $t$:
\begin{align}
\frac{\partial}{\partial t}\mathbb{E}\left(Z_t\mathbb{1}_{T_0>t}\conditional~{} d_0\right)&\approx2\Delta\lambda\mathcal{M}_{Z_t}(0)\\
\frac{\partial}{\partial t}\mathbb{E}\left(Z_t^2\mathbb{1}_{T_0>t}\conditional~{} d_0\right)&\approx4\Delta\lambda\left(\mathcal{M}_{Z_t}(0)+2\mathbb{E}\left(Z_t\mathbb{1}_{T_0>t}\conditional~{} d_0\right)\right)
\end{align}
Furthermore, Lemma~\ref{lemma:convergence}\textit{b)} states that $\mathcal{M}_{Z_t}(0)$ can be treated like a nonzero constant for large $t$. This allows us to solve the system exactly (substituting equalities for both "$\approx$"): 
\begin{align}
\mathbb{E}\left(Z_t\mathbb{1}_{T_0>t}\conditional~{} d_0\right)&\approx2\Delta\lambda\mathcal{M}_{Z_t}(0)\cdot t+c_1\\
\mathbb{E}\left(Z_t^2\mathbb{1}_{T_0>t}\conditional~{} d_0\right)&\approx 8\Delta ^2\lambda ^2 \mathcal{M}_{Z_t}(0)\cdot t^2+4\Delta\lambda(\mathcal{M}_{Z_t}(0)+2c_1)\cdot t+c_2
\end{align}
with initial values $c_1,c_2>0$ (for which we could use $c_1=\mathbb{E}\left(Z_0\mathbb{1}_{T_0>0}\conditional~{} d_0\right)=d_0/4\theta^2$ and $c_2=\mathbb{E}\left(Z_0^2\mathbb{1}_{T_0>0}\conditional~{} d_0\right)=\left(\frac{d_0}{4\theta^2}\right)^2$).

Let $\mu_1(t):=\mathbb{E}\left(Z_t\conditional~~ T_0>t,d_0\right),\mu_2(t):=\mathbb{E}\left(Z_t^2\conditional~~ T_0>t,d_0 \right)$ and $\sigma^2(t)$ denote the first and second moment under the condition of no coalescence up to time $t$, and the conditional variance respectively. We obtain these moments by dividing the above equations by $\Pr(T_0>t\conditional~{} d_0)=\mathcal{M}_{D_t}(0)$. From the above approximation, we get:

\begin{align}
\label{eq:mut}\mu_1(t)&\approx2\Delta\lambda\cdot t+\frac{c_1}{\mathcal{M}_{Z_t}(0)}\\
\label{eq:vart}\mu_2(t)&\approx 8\Delta ^2\lambda ^2 \cdot t^2+4\Delta\lambda\left(1+2\frac{c_1}{\mathcal{M}_{Z_t}(0)}\right)\cdot t+\frac{c_2}{\mathcal{M}_{Z_t}(0)}
\end{align}
$\sigma^2_t$ is obtained by applying $\mathbb{V}(X)=\mathbb{E}(X^2)-\mathbb{E}(X)^2$. These approximations are not very precise and should only be taken to reflect the asymptotic behaviour. Simulations suggest that their accuracy increases considerably if the initial distance $d_0$ is taken to be large (see Figure~\ref{fig:avgvar}). In any case, if $t$ is large, $\mathcal{M}_{Z_t}(0)$ is close to $p^*$ (Lemma~\ref{lemma:convergence}\textit{b}), so $c_1/\mathcal{M}_{Z_t}(0)$ and $c_2/\mathcal{M}_{Z_t}(0)$ are almost constant, which shows that $\mu_1(t)$ is asymptotically linear (and $\mu_2(t)$ quadratic).

\begin{figure}

\begin{subfigure}[(G1)]{\textwidth}
\includegraphics[scale=.75]{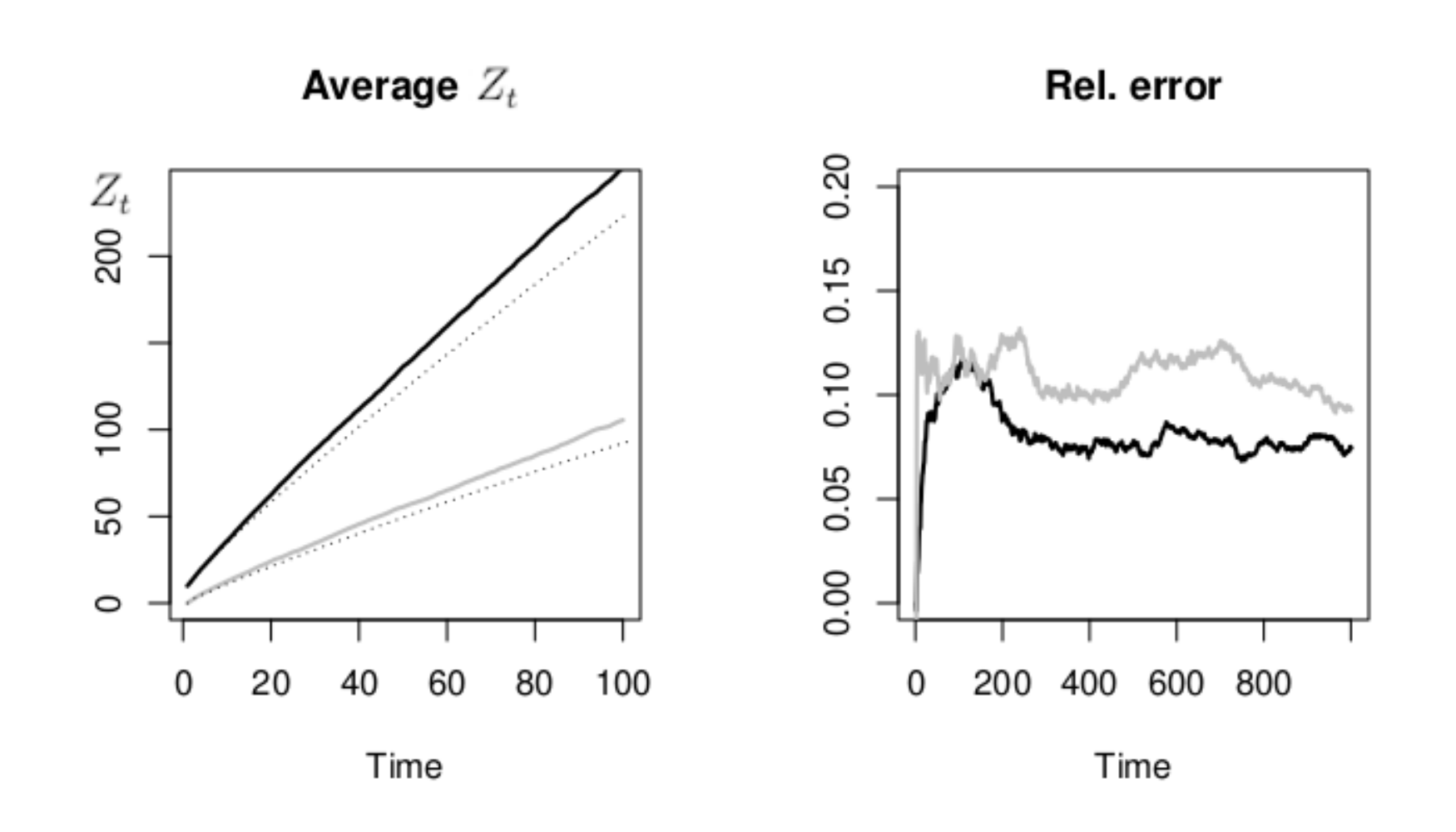}
\end{subfigure}
\begin{subfigure}[(G2)]{\textwidth}
\includegraphics[scale=.75]{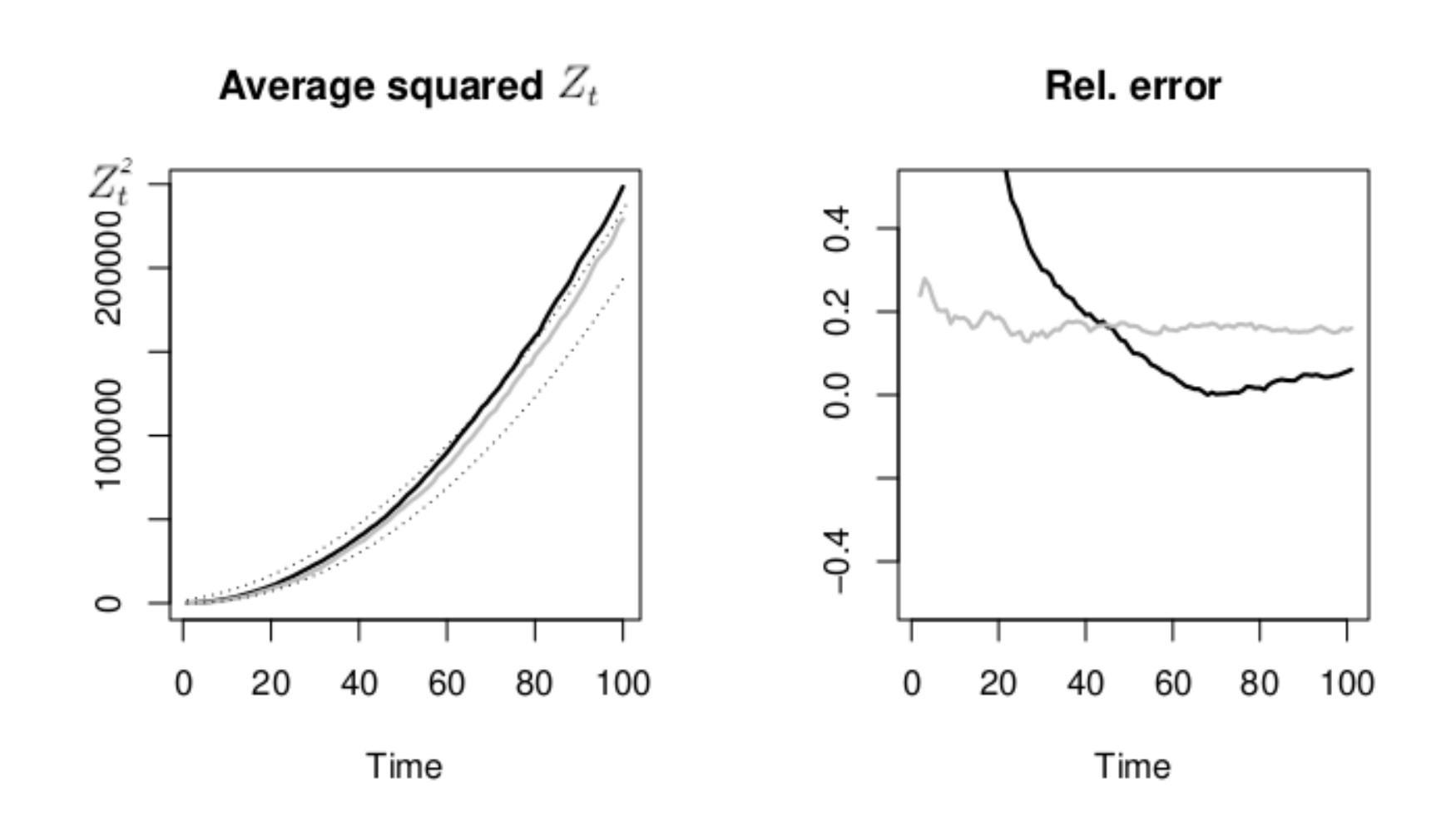}
\end{subfigure}
\caption{\textit{Top left:} Average of $Z_t$ (i.e., $\mu_1(t)$) with respect to time and $d_0=0$ (grey) and $d_0=10$ (black). Shown in dotted lines is the respective approximation obtained by the results of section~\ref{sec:asymp}. \textit{Top right:} One observes that the relative error is smaller for $d_0=10$ than for $d_0=0$. 
\textit{Bottom:} Here, we compare the average squared $Z_t$ (i.e., $\mu_2(t)$) with its approximation. Again, $d_0=10$ (as $t$ increases) yields the smallest relative error. }
\label{fig:avgvar}
\end{figure}
\subsection{Using the characteristic function of the Gamma Distibution}\label{sec:gamma}
In the following, we will consider the random variable $\mathcal{Z}_t$, defined as $Z_t$ conditioned on $T_0>t$. Consequently, $\mathbb{E}\left(\mathcal{Z}_t\conditional~{} d_0\right)=\mu_1(t)$, $\mathbb{E}\left(\mathcal{Z}_t^2\conditional~{} d_0\right)=\mu_2(t)$ and $\mathbb{V}\left(\mathcal{Z}_t\conditional~{} d_0\right)=\sigma^2(t)$.\\
We assume that $\mathcal{Z}_t$ follows a Gamma distribution with shape and rate parameters $\alpha_t$ and $\beta_t$. This assumption is justified by the fact that the distribution of $\mathcal{Z}_t$ is a convolution of $\chi^2$-distributions, which are of the Gamma family, and can be supported by simulations (see Figure~\ref{fig:gamma}). Then, we have for the characteristic function of $\mathcal{Z}_t$, evaluated at $s=i$ ($i$ representing the imaginary unit):
\begin{eqnarray}
\chi_{\Gamma(\alpha_t,\beta_t)}(i)=\mathbb{E}\left(\exp(-\mathcal{Z}_t) \conditional~{} d_0\right) = \left(1+\frac{1}{\beta_t}\right)^{-\alpha_t}
\end{eqnarray}

The parameters $\alpha_t$ and $\beta_t$ can be computed from the first two moments of the distribution:
\begin{align*}
\beta_t&=\mu_1(t)\left(\sigma^2(t)\right)^{-1}\\
\alpha_t&=\mu_1(t)^2\left(\sigma^2(t)\right)^{-1}
\end{align*}
\begin{figure}
    \centering
    \includegraphics[scale=.625]{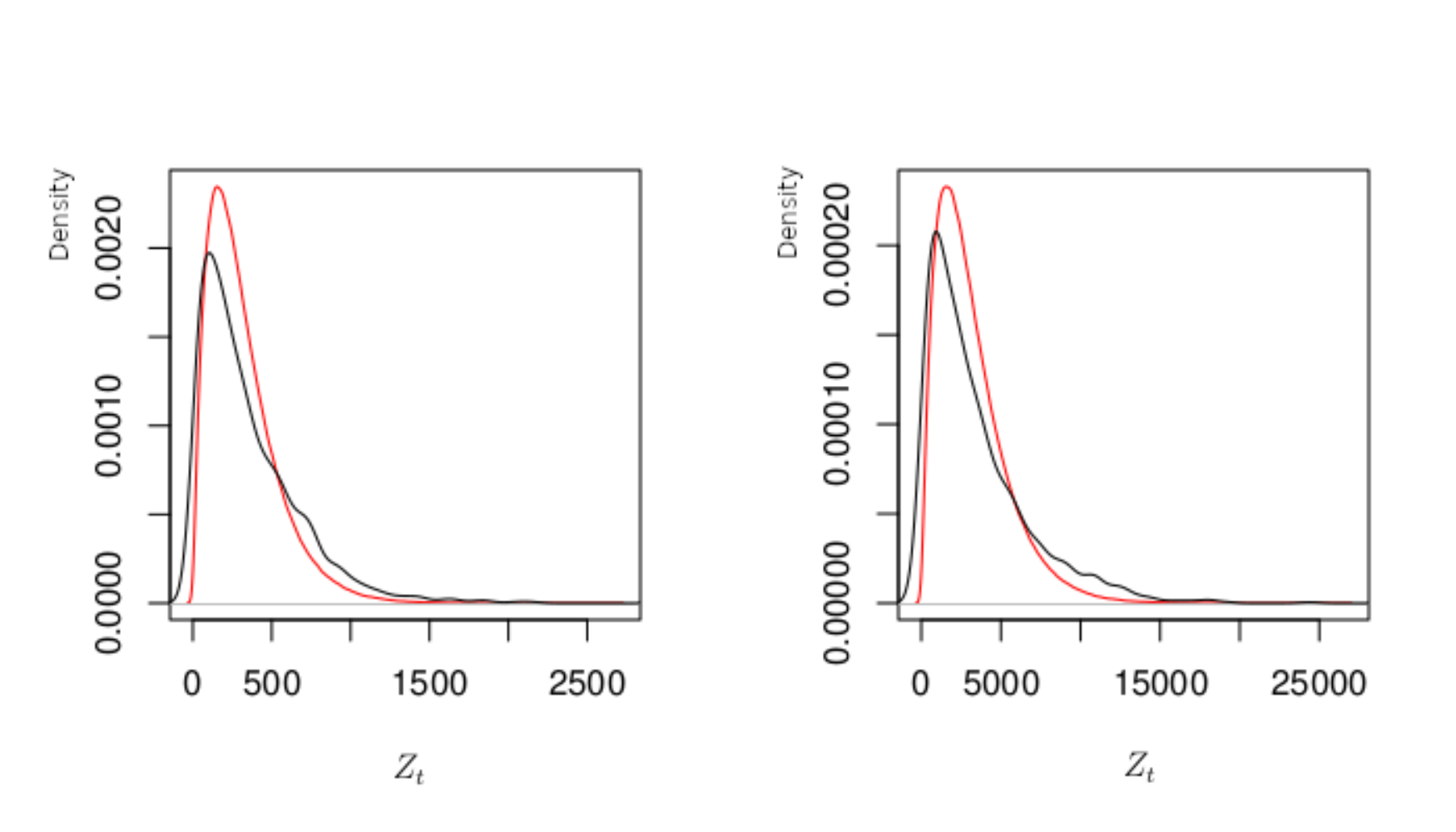}
    \caption{Gamma distributions (red) with parameters derived from the asymptotic moments are good approximates for the observed distributions (black) of $\mathcal{Z}_t$. \textit{Left}: $t=100$, \textit{right}: $t=1000$.}
    \label{fig:gamma}
\end{figure}
\eq{eq:coalrate-ode} then becomes
\begin{equation}\label{eq:coalrate-cfgamma}
\Pr(T_0\leq t|d_0)\approx 1-\exp\left(-\int_0^{t}\frac{1}{2}\Delta \lambda \left(1+\frac{\sigma^2(u)}{\mu_1(u)}\right)^{-\frac{\mu_1(u)^2}{\sigma^2(u)}} \mathrm{d}u\right)
\end{equation}
In this, we can use the approximations for $\mu_1(t)$ and $\sigma^2(t)$ (\eq{eq:mut}, \eq{eq:vart}) developed in section~\ref{sec:asymp}. Equipped with this, \eq{eq:coalrate-cfgamma} is an approximation scheme for $\Pr(T_0\leq t|d_0)$.\\
While this usually is not particularly close to the distribution of $T_0$ (see Figure~\ref{fig:smoothunsmooth}), there are several possibilities of improvement. If there are known or reasonably well approximated values of $\Pr(T_0\leq x|d_0)$ for some $x>0$, we have
\begin{equation}\label{eq:coalrate-cfgamma-init}
\Pr(T_0\leq t|d_0)\approx 1-\Pr(T_0> x|d_0)\exp\left(-\int_x^{t}\frac{1}{2}\Delta \lambda \left(1+\frac{\sigma^2(u)}{\mu_1(u)}\right)^{-\frac{\mu_1(u)^2}{\sigma^2(u)}} \mathrm{d}u\right)
\end{equation}
For example, $\mathcal{M}_{Z_t}(0)$ can be approximated up to $x>0$ by a Taylor polynomial, while for $t>x$ one utilizes \eq{eq:coalrate-cfgamma-init}. Below, we show the result of this using the asymptotic approximations (\eq{eq:mut}, \eq{eq:vart}) for the moments (this procedure is dubbed the "naive" approach). It is possible to attain more precision by calculating $\mu_1(x)$ and $\sigma^2(x)$ exactly (e.g., by another Taylor scheme), and continue $\mu_1(t)$ linearly and $\sigma^2(t)$ quadratically for $t>x$, using the results of section~\ref{sec:asymp}.\\
These approximation schemes typically results in a "knee" of the curve (i.e, a point where it visibly ceases to be smooth). One can mitigate this by calculating the (unique) values of $\tilde{\mu}_1(x)$ and $\tilde{\sigma}^2(x)$ that provide a smooth continuation in \eq{eq:coalrate-cfgamma-init} (they do not necessarily equal the true values $\mu_1(x), \sigma^2(x)$). For $t>x$, $\mu_1(t)$ can be extended linearly and $\sigma^2(t)$ quadratically. This "smooth" way of approximating $\mathcal{M}_{Z_t}(0)$ fits the distribution of $T_0$ rather well. Generally, a higher threshold $x$ results in a higher accuracy.
\subsection{Utilizing the similarity of the slope of $\mathcal{M}_{Z_t}(0)$ for differing $d_0$}
In practice, it may become necessary to approximate $\mathcal{M}_{Z_t,d_0}(0)$ for several different initial values of $d_0$ (As a reminder to the reader, $d_0$ is included in the subscript in the definition of $\mathcal{M}_{Z_t,d_0}(0)$, see~\eq{eq:mdef}, but was omitted throughout most of the previous sections for the sake of convenience). In order to do that, one can take advantage of the observation in Figure~\ref{fig:coalprob} that in the long run, the derivatives $\frac{\partial}{\partial t}\mathcal{M}_{Z_t,d_0}(0)$ for different initial conditions start to closely resemble each other; i.e., their slopes become similar as $t$ gets large. One possible explanation of this is that lineage pairs are expected to grow apart almost linearly, regardless of the value of $d_0$, unless they coalesce early on. To support this claim, we refer to our simulations.\\
Assuming we have a precomputed approximation $M_{0}(t)$ of $\mathcal{M}_{Z_t,0}(0)$, a fast way of approximating $\mathcal{M}_{Z_t,d}(0)$ for some $d>0$ is therefore the utilization of the Taylor expansion for small values of $t\leq x$, and continuing the curve for $t>x$ by the slope of $\mathcal{M}_{Z_t,0}(0)$ that is obtained according to the scheme we discussed (Section~\ref{sec:gamma}). Again, this can be done in such a way that the resulting curve is smooth, e.g., by "shifting" to the point $x'$ that guarantees a smooth transition. To be precise, let $\gamma_{d}^{(k)}(t)$ denote a Taylor approximation of $\mathcal{M}_{Z_t,d}(0)$. Then, for given $x>0$, we may choose $x'$ such that the function
\begin{equation}
M_{d}(t)=
\begin{cases}
\gamma_{d}^{(k)}(t) & t\leq x \\
M_{0}(t-x+x') & t>x \\
\end{cases}
\end{equation}
is smooth. Hence, $M_{d}(t)$ is an approximation of $\mathcal{M}_{Z_t,d}(0)$ whose long-term behaviour agrees with that of $M_{0}(t)$.
\subsection{Visualisations of the approximations}
We show some approximations of the coalescent probability for $\theta^2,\lambda=1,u_0=1$ and $Z_0=0$. The density itself is obtained from a set of 10000 simulation runs of $Z_t$ (see also Figures~\ref{fig:trajectories} and~\ref{fig:coalprob}).
\begin{figure}
\includegraphics{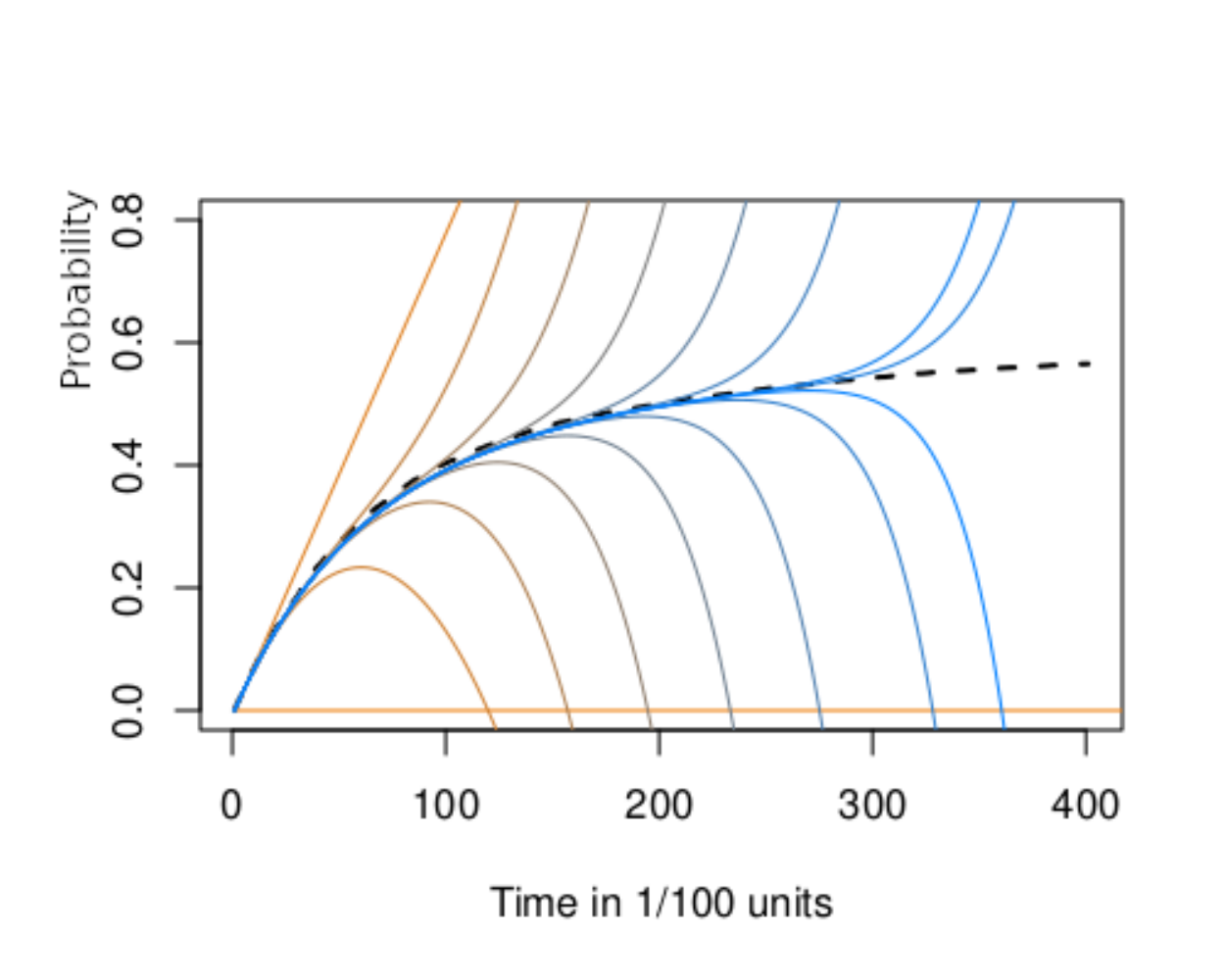}
\caption{Taylor approximations of the probability density of coalescences. Shown in color from yellow to blue are the approximations of increasing order up to 16.}
\label{fig:taylor}
\end{figure}
\begin{figure}
\begin{subfigure}[(G1)]{0.5\textwidth}
\includegraphics[scale=.5]{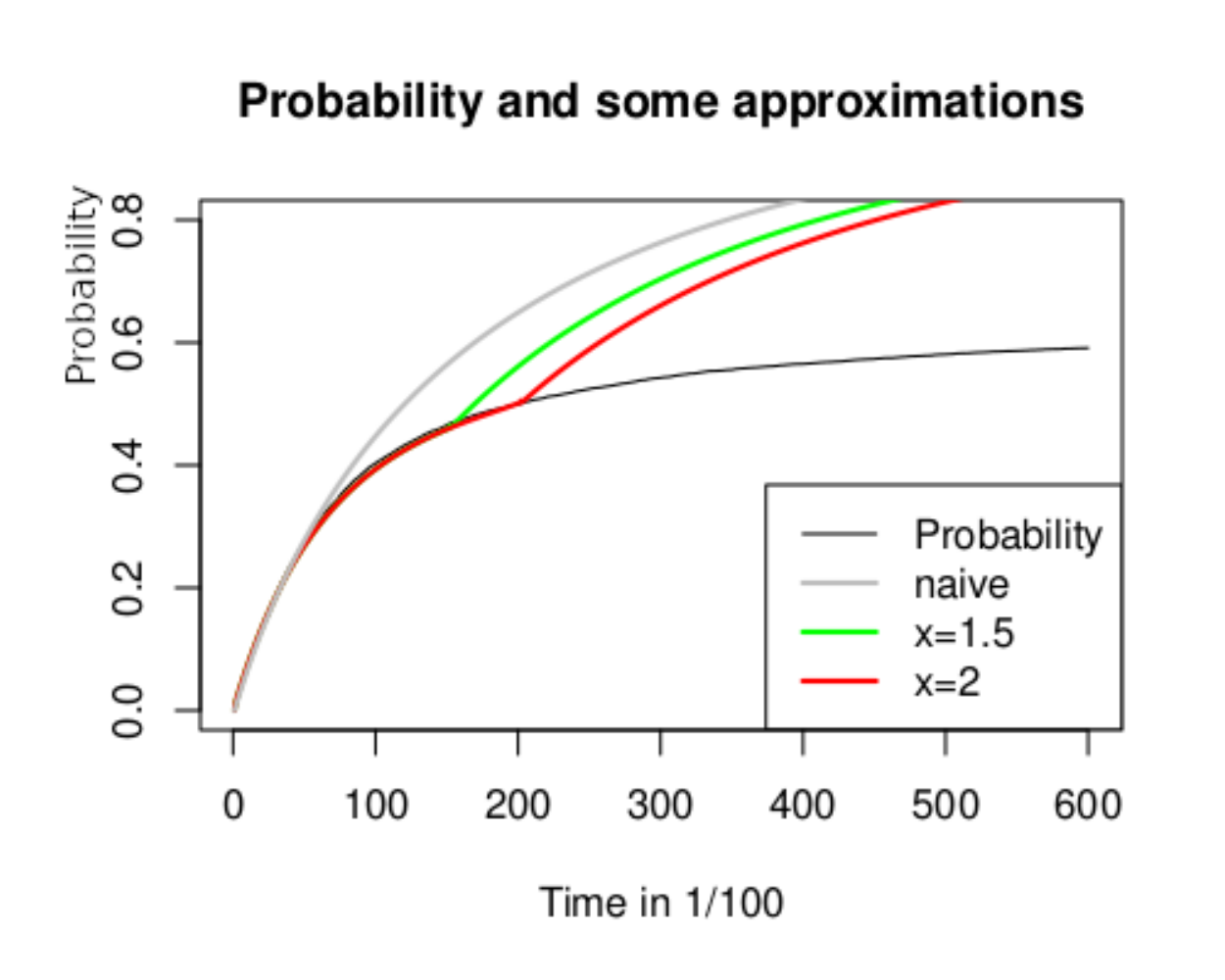}
\end{subfigure}
\begin{subfigure}[(G2)]{0.5\textwidth}
\includegraphics[scale=.5]{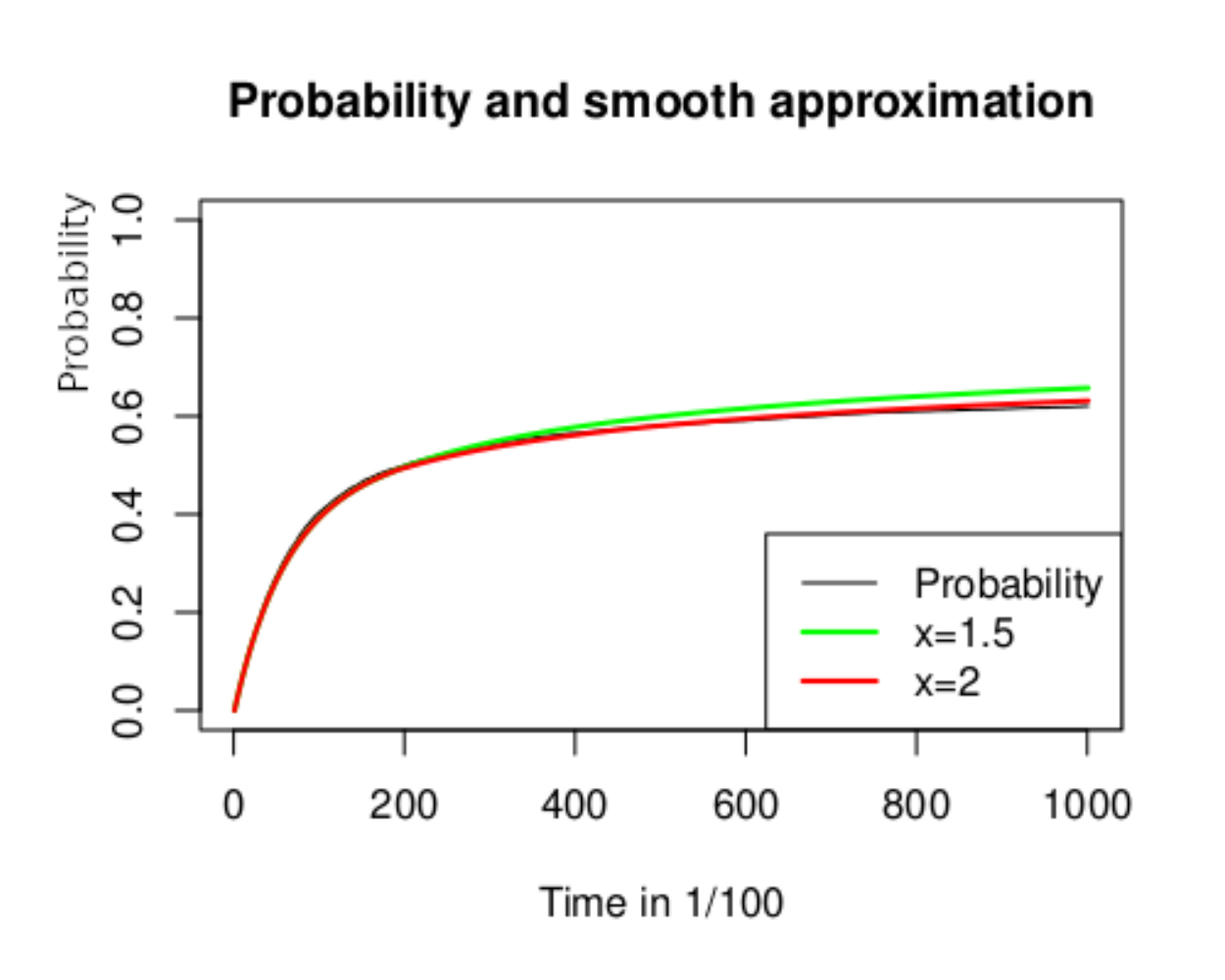}
\end{subfigure}
\caption{Comparison between the "naive" approximation using \eq{eq:coalrate-cfgamma} as well as continuing with \eq{eq:coalrate-cfgamma-init} after the threshold value $x$ (\textit{left}), and the "smooth" approach. Not only does the naive way introduce an unrealistic point of non-differentiability, but also deviates from the probability distribution much faster.}
\label{fig:smoothunsmooth}
\end{figure}
\begin{figure}

\includegraphics{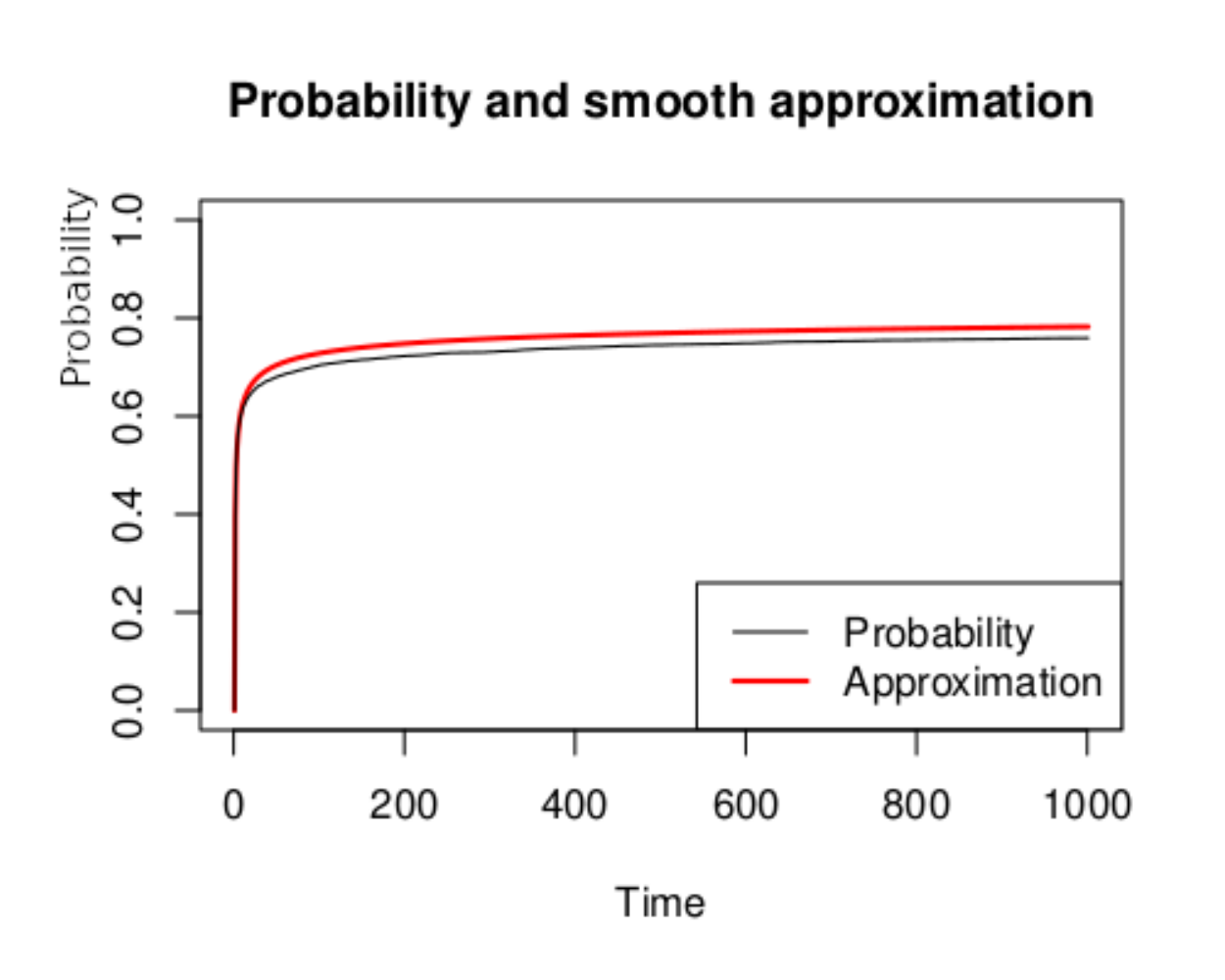}

\caption{Long-term behaviour of the probability distribution and the smooth approximation for $x=2$.}
\label{fig:approx}
\end{figure}

\section{Discussion}
In this study, we  describe ways of approximating the distribution of the time to coalescence $\Pr(T_0\leq t|d_0)$ for pairs of lineages under the \slfv~on $\mathbb{R}^2$ as well as on a finite rectangle. The major difference between the two cases is the fact that $T_0$ is almost surely finite on a rectangle, whereas our analysis showed that on the plane, there is a chance $p^*$, dependent on the parameter choice of the model, that lineage pairs escape the coalescent mechanism entirely (as has been suspected and investigated before in different contexts, e.g. \cite{veber:slfv}). It is noteworthy that this does not necessarily pose a problem in applications, since $T_0$ can be conditioned on coalescence taking place by simply dividing its distribution by $1-p^*$.\\
For the \slfv~on a rectangle $\mathcal{A}$, $\Pr(T_0\leq t|x_0,y_0)$ becomes proportional to an exponential function as $t$ gets large. It is reasonable to expect this to be true not just on a rectangle, but on any  compact habitat. Especially interesting cases for similar analyses appear to be spherical and toric habitats, since those are compact objects, but the reproduction mechanism can be defined such that border effects are avoided. Regarding the approximation of $\Pr(T_0\leq t|x_0,y_0)$, there are two cases to be distinguished: First, if $\mathcal{A}$ is of "moderate" size in relation to $\theta^2$, the determining parameters of the distribution can be well-approximated, as the equilibrium distribution of the location of the two lineages is nearly uniform and the corresponding density for the distribution of pairwise distance has a closed-form formula \cite{philip2007}. If $\mathcal{A}$ is large, simulations show that the deviation of the equilibrium distribution, while small in absolute terms, negatively impacts the approximation scheme. However, one natural solution to this problem might be to simply rely on the results for the process on $\mathbb{R}^2$ in such a case.\\
On $\mathbb{R}^2$, we related $\Pr(T_0\leq t|d_0)$ to the process of $Z_t$, which was defined as the squared euclidean distance $D_t$ between a pair of lineages, divided by $4\theta^2$. Our analysis revealed several properties of $Z_t$ as well, for instance the asymptotically linear growth of its expectation. The function $\mathcal{M}_{Z_t}(s)$, defined similarly to a moment generating function, with the addition of an indicator variable, is linked to $\Pr(T_0\leq t|d_0)$ via a series of algebraic and differential equations. Perhaps similar techniques can be used to analyze survival times in other stochastic processes with spontaneous entries into absorbing states (for the \slfv, these are the coalescences). The approximation scheme we suggest relies on evaluating $\Pr(T_0\leq t|d_0)=1-\mathcal{M}_{Z_t}(0)$ for small $t$ using a Taylor polynomial, and continuing the approximation for larger $t$ by assuming that $\mathcal{Z}_t$ is gamma-distributed.\\
We have ignored  $u_0$, the "mortality" parameter, by systematically considering its value as equal to 1 in our derivations. Intuitively, the effect of changing $u_0$ should be similar to changing the value of $\lambda$. Calculations presented in Section \ref{sec:plane} can be repeated by making the parameter $u_0$ explicit, resulting in an additional parameter in the approximation. Yet, more work would be required in order to verify that coalescence probabilities that involve $u_0$ could eventually be approximated by using the very same approaches as the ones presented here.\\
The presented results describe the ancestral process for a sample of size $n=2$ under the \slfv. For bio-statistical purposes, it would certainly be helpful to extend these to bigger sample sizes. Most of the formulae in Section \ref{sec:dynamics} can, at first glance, be modified to incorporate more than two lineages, so it seems possible to approach waiting times for multiple mergers at least approximately, which would allow for more sound statistical analyses in phylogeography. However,  knowledge about the process for $n=2$ enables statistical assessment at least for pairs of samples, with all pairs in a sample considered as independent from one another. We maintain that the presented methodology may be used to obtain estimates of the parameters $\theta^2, \lambda$ from geo-referenced genomic data. At the very least, estimates of $\theta^2$ and $\lambda$ obtained in this way may serve as valuable indicators of the speed of evolution and dispersing potential of biological organisms.\\
\section*{Acknowledgements}
This work was funded by the Agence Nationale
pour la Recherche [\url{https://anr.fr/}] through the grant GENOSPACE, and the Walter-Benjamin Program (WI 5589/1-1) of the DFG [\url{https://dfg.de/}].

\section{Appendix}
\begin{proof}[Derivation of \eq{eq:coalincorporated}]
We have
\begin{align}\label{eq:coalincorporated2}
&\ghosteq\frac{\partial}{\partial t}\mathcal{M}_{Z_t}(s)\\
\notag&=\lim_{h\rightarrow 0}\frac{1}{h}\left(\mathbb{E}\left(\exp\left(-sZ_{t+h}\right)\mathbb{1}_{T_0>t+h}\conditional~{} d_0\right)-\mathbb{E}\left(\exp\left(-sZ_t\right)\mathbb{1}_{T_0>t}\conditional~{} d_0\right)\right)\\
\notag&=\mathbb{E}\left(\lim_{h\rightarrow 0}\frac{1}{h}\mathbb{E}\left(\exp\left(-sZ_{t+h}\right)\mathbb{1}_{T_0>t+h}-\exp\left(-sx\right) \conditional~{} Z_t=x, d_0\right)\mathbb{1}_{T_0>t}\conditional~{} d_0\right)
\end{align}
We may decompose the interior expectation by conditioning the number $N_h$ of events affecting the lineages that are encountered in the interval $[t,t+h]$. Any event that hits either $X$ or $Y$ or both contributes to the number $N_h$. We will show that in the limit of $h\rightarrow 0$, only the case $N_h=1$ is relevant.

\begin{align}\label{eq:coalincorporated3}
\notag &\ghosteq\mathbb{E}\left(\exp\left(-sZ_{t+h}\right)\mathbb{1}_{T_0>t+h}-\exp\left(-sx\right) \conditional~{} Z_t=x, d_0\right)\\
 &= \sum_{i\in\mathbb{N}}\Pr(N_h=i\conditional~{} Z_t=x, d_0)\mathbb{E}\left(e^{-sZ_{t+h}}\mathbb{1}_{T_0>t+h}-e^{-sx}\conditional~{} N_h=i,Z_t=x,d_0\right)
\end{align}
If $N_h=0$, we have $Z_{t+h}=x$ and the corresponding term vanishes. For $N_h=1$, we can express the probability by multiplying the density with which an event occurs at $t+u\in[t,t+h]$ (exponential with parameter given by \eq{eq:totalrate}) with the probability of no further event during the remainder of this interval 
$$2\Delta\lambda\left(1-\frac{e^{-Z_t}}{4}\right)e^{-2\Delta\lambda\left(1-\frac{e^{-Z_t}}{4}\right)u}\cdot e^{-2\Delta\lambda\left(1-\frac{e^{-Z_{t+u}^+}}{4}\right)(h-u)}$$
and integrating over $u\in[0,h]$. $Z_{t+u}^+$ denotes the distance of the lineages immediately after the event that occurs at time $t+u$; note that $X$ and $Y$ may coalesce due to this event, in which case $Z_{t+u}^+=0$. Then, it holds that
\begin{align}\label{eq:eventintensity}
\notag&\ghosteq\lim_{h\rightarrow 0}\frac{1}{h}\Pr(N_h=1\conditional~{} Z_t=x, d_0)\\
\notag&=\lim_{h\rightarrow 0}\frac{1}{h}\int_0^h2\Delta\lambda\left(1-\frac{e^{-Z_t}}{4}\right)e^{-2\Delta\lambda\left(1-\frac{e^{-Z_t}}{4}\right)u+2\Delta\lambda\left(1-\frac{e^{-Z_u^+}}{4}\right)(h-u)}du\\
&=2\Delta\lambda\left(1-\frac{e^{-Z_t}}{4}\right)
\end{align}
As for $N_h\geq 2$, it clearly holds that
\begin{equation}
\Pr(N_h=k\geq 2\conditional~{} Z_t=x,d_0)\leq \left(2\Delta\lambda h\right)^k/k!e^{-2\Delta\lambda}
\end{equation}
because $2\Delta\lambda$ is an upper bound to the total rate of events (\eq{eq:totalrate}). Because of that, we have 
$$\lim_{h\rightarrow 0}\frac{1}{h}\Pr(N_h=k\geq 2\conditional~{} Z_t=x)\leq \lim_{h\rightarrow 0}\frac{1}{h}\frac{(2\Delta\lambda)^k}{k!}\exp(-2\Delta\lambda) =0$$
Close inspection of the term corresponding to $N_h=1$ reveals
\begin{align}\label{eq:lineagejumps}
\notag &\ghosteq \lim_{h\rightarrow 0}\mathbb{E}\left(2\Delta\lambda\left(1-\frac{e^{-Z_t}}{4}\right)e^{-\frac{sZ_{t+h}}{4\theta^2}}\mathbb{1}_{T_0>t+h}-e^{-\frac{sx}{4\theta^2}}\conditional~{} N_h=1,Z_t=x,d_0\right)\\
\notag &=2\Delta\lambda\mathbb{E}\left(\left(1-\frac{e^{-x}}{4}\right)e^{-\frac{sZ_{t}^+}{4\theta^2}}\mathbb{1}_{T_0>t}-e^{-\frac{sx}{4\theta^2}}\conditional~{} Z_t=x,d_0,\textnormal{ Event at }t\right)\\
&=2\Delta\lambda\mathbb{E}\left(\int\left(1-\frac{e^{-x}}{2}\right) \left(e^{-\frac{sw}{4\theta^2}}-e^{-\frac{sx}{4\theta^2}}\right)p_{\{X_{t}^+\conditional~{}!Y_t\}}(w)\mathrm{d}w\conditional~{} Z_t=x,d_0\right)\\
\notag &-\Delta\lambda\frac{e^{-x/4\theta}}{2}e^{-\frac{sZ_t}{4\theta^2}}
\end{align}
where the expression~(\ref{eq:lineagejumps}) accounts for all cases in which only one lineage is affected by the event, and the one below for those in which the lineages coalesce. Let $w$ denote the value $Z_t^+$ immediately after the event. We can assume without restriction that lineage $X_t$ is hit by the event and $Y_t$ remains at its position. Then, the density of $w$ is given by~\eq{eq:nocoal-conditional}. We find
\begin{align}\label{eq:mgf-conditional}
&\ghosteq \int\left(1-\frac{e^{-x}}{2}\right) \left(e^{-\frac{sw}{4\theta^2}}-e^{-\frac{sx}{4\theta^2}}\right)p_{\{X_{t}^+\conditional~{}!Y_t\}}(w)\mathrm{d}w\\
&=\frac{1}{s+1}e^{-\frac{\frac{s}{1+s}x}{4\theta^2}}-e^{-\frac{sx}{4\theta^2}}-\frac{2}{3s+4}e^{-\frac{\frac{4(s+1)}{3s+4}x}{4\theta^2}}+\frac{1}{2}e^{-\frac{(s+1)x}{4\theta^2}}
\end{align}
The evaluation of the integral is extensive, but ultimately trivial. Assembling everything, we arrive at \eq{eq:coalincorporated}:
\begin{align}
\notag&\ghosteq\frac{\partial}{\partial t}\mathcal{M}_{Z_t}(s)\\
\notag&=\mathbb{E}\left(2\Delta\lambda\left(\frac{1}{s+1}e^{-\frac{s}{1+s}Z_t}-e^{-sZ_t}+\frac{2}{3s+4}e^{-\frac{4(s+1)}{3s+4}Z_t}\right)\mathbb{1}_{T_0>t}\conditional~{} d_0\right)\\
&+\mathbb{E}\left(\frac{\Delta\lambda}{2} e^{-(1+s)Z_t}\mathbb{1}_{T_0>t}\conditional~{} d_0\right)
\end{align}
from which the claimed identity follows by linearity of the expectation. 
\end{proof}
\begin{proof}[Proof of Lemma~\ref{lemma:convergence}]
\textit{a)} Since $\mathcal{M}_{Z_t}(0)\in[0,1]$ and monotonously falling, there exists a limit $c\in[0,1]$. Consequently, $\lim_{t\rightarrow\infty}\mathcal{M}_{Z_t}(1)=0$, because of the established relationship between the two (\eq{eq:delp}).\\
For $s>1$, $\mathcal{M}_{Z_t}(1)>\mathcal{M}_{Z_t}(s)$, so $\lim_{t\rightarrow\infty}\mathcal{M}_{Z_t}(s)=0$ holds here as well. Concerning $s<1$, we have
$$\mathcal{M}_{Z_t}(s)=\mathbb{E}\left(\exp\left(-sZ_t\right)\mathbb{1}_{T_0>t}\conditional~{} d_0\right)\leq\mathbb{E}\left(\exp\left(-Z_t\right)\mathbb{1}_{T_0>t}\conditional~{} d_0\right)^{s}=\left(\mathcal{M}_{Z_t}(1)\right)^s$$
by Markov's inequality. Therefore, $\mathcal{M}_{Z_t}(s)$ converges to $0$ uniformly on any interval $[\sigma,\infty),\sigma>0$.\\
Having established that $\mathcal{M}_{Z_t}(s)$ converges to $0$, we know that for any $\delta\in\mathbb{R}^+$, the probability $\Pr(Z_t\leq \delta,T_0>t\conditional~{} d_0)$ converges to zero as well (otherwise, we end up with a contradiction). Consequentially, since $\left(\frac{x}{4\theta}\right)^k\exp\left(-\frac{sx}{4\theta^2}\right)$ converges to $0$ for any $k>0$ as $x\rightarrow\infty$, all the derivatives, given by $\left(\frac{\partial}{\partial s}\right)^k\mathcal{M}_{Z_t}(s)=\mathbb{E}\left(\left(-Z_t\right)^k\exp\left(-sZ_t\right)\mathbb{1}_{T_0>t}\conditional~{} d_0\right)$ necessarily converge to $0$ as well.\\
\textit{b)} Consider the function $y(t):=\mathbb{E}\left(\frac{Z_t}{t}\mathbb{1}_{T_0>t}\conditional~{} d_0\right)$. It solves the differential equation
\begin{align}
\notag\frac{\partial}{\partial t}y(t)&=\left(2\Delta\lambda\left(\mathcal{M}_{Z_t}(0)-\frac{1}{8}\left(3\mathcal{M}_{Z_t}(1)+\frac{\partial}{\partial s}\mathcal{M}_{Z_t}(s)\conditional~{}_{s=1}\right)\right)\right)\cdot t^{-1}\\
&\ghosteq-\mathbb{E}\left(\frac{Z_t}{t}\mathbb{1}_{T_0>t}\conditional~{} d_0\right)\cdot t^{-1}\\
\notag&=\left(2\Delta\lambda\left(\mathcal{M}_{Z_t}(0)-\frac{1}{8}\left(3\mathcal{M}_{Z_t}(1)+\frac{\partial}{\partial s}\mathcal{M}_{Z_t}(s)\conditional~{}_{s=1}\right)\right)\right)\cdot t^{-1}\\
&\ghosteq-y(t)\cdot t^{-1}
\end{align}
Since the first term on the right-hand side is strictly positive (we have $\mathcal{M}_{Z_t}(1)\leq\mathcal{M}_{Z_t}(0)$, $\frac{\partial}{\partial s}\mathcal{M}_{Z_t}(s)\conditional~{}_{s=1}\leq\mathcal{M}_{Z_t}(0)$), the following must hold for any solution $y(t)$ to this differential equation:
\begin{equation}
\forall t>t_0:y(t)>z(t) \textnormal{ if } \frac{\partial}{\partial t}z(t)=z(t)\cdot t^{-1}
\end{equation}
provided some initial value $c=y(t_0)=z(t_0),t_0>0$. Moreover, if $\lim_{t\rightarrow\infty}z(t)=\delta\in\mathbb{R}$, then $\lim_{t\rightarrow\infty}y(t)>\delta$.
The solution to $\frac{\partial}{\partial t}z(t)=z(t)\cdot t^{-1},z(t_0)=c$ is $z(t):=ct_0\cdot t^{-1}$ with limit $0$ as $t\rightarrow\infty$. Thus,
\begin{equation}
\lim_{t\rightarrow\infty}y(t)=\lim_{t\rightarrow\infty}\mathbb{E}\left(\frac{Z_t}{t}\mathbb{1}_{T_0>t}\conditional~{} d_0\right)>0
\end{equation}
Therefore, there exists an $\alpha>0$ such that $\mathbb{E}\left(\frac{D_t}{4\theta^2}\mathbb{1}_{T_0>t}\conditional~{} d_0\right)>\alpha t$ for all $t$ larger than some $t_0$. In turn, looking back at \eq{eq:dtgdl}, this implies that $\Pr(T_0>t\conditional~{} d_0)=\mathbb{E}(\mathbb{1}_{T_0>t}\conditional~{} d_0)=\mathcal{M}_{Z_t}(0)$ does not converge to $0$.
\end{proof}

\bibliography{references}

\end{document}